\documentclass[a4paper,10pt,twocolumn]{elsarticle}
 \makeatletter
 \def\ps@pprintTitle{%
  \let\@oddhead\@empty
  \let\@evenhead\@empty
  \def\@oddfoot{}%
  \let\@evenfoot\@oddfoot}
 \makeatother

 \usepackage{hyperref}
 \usepackage{amsmath}
 \usepackage{amssymb}
 \usepackage{tikz}

 \newdefinition{definition}{Definition}
 
 \newtheorem{theorem}{Theorem}
 \newtheorem{lemma}[theorem]{Lemma}
 
 \newtheorem{proposition}[theorem]{Proposition}
 \newproof{proof}{Proof}

 \usepackage{geometry}
  \usepackage{multirow}

  \usepackage{graphicx}

  \newcommand{\R}{\mathbb{R}}
  \newcommand{\eps}{\epsilon}
  \newcommand{\terrain}{T}
  \newcommand{\terrainheight}{t}
  \newcommand{\altterrain}{T'}
  \newcommand{\altterrainheight}{t'}
  \newcommand{\simpterrain}{S}
  \newcommand{\baseterrain}{B}
  \newcommand{\simpterrainheight}{s}
  \newcommand{\baseterrainheight}{b}
  \newcommand{\linktri}{L}

  \DeclareMathOperator{\link}{Lk}

\begin{document}

\twocolumn[{
\begin{frontmatter}

\title{Topology-Preserving Terrain Simplification}

\author[mymainaddress]{Ulderico Fugacci}

\author[mysecondaryaddress]{Michael Kerber}

\author[mythirdaddress]{Hugo Manet}

\address[mymainaddress]{IMATI - CNR, Genova, Italy}
\address[mysecondaryaddress]{Institute of Geometry, TU Graz, Graz, Austria}
\address[mythirdaddress]{D\'epartement d'Informatique, \'Ecole Normale Sup\'erieure, Paris, France}

\begin{abstract}
We give necessary and sufficient criteria for elementary operations in a two-dimensional terrain to preserve the persistent homology induced by the height function. These operations are edge flips and removals of interior vertices, re-triangulating the link of the removed vertex. This problem is motivated by topological terrain simplification, which means removing as many critical vertices of a terrain as possible while maintaining geometric closeness to the original surface. Existing methods manage to reduce the maximal possible number of critical vertices, but increase thereby the number of regular vertices. Our method can be used to post-process a simplified terrain, drastically reducing its size and preserving its favorable properties.
\end{abstract}



\end{frontmatter}
}]


\section{Introduction}
\label{sec-introduction}

A terrain is a triangulated surface described by a scalar function defined on a finite set of points of $\R^2$ and visualized as a height function of some topographic dataset, representing the mountains and valleys of a landscape. Terrains are a popular model to represent landscapes
and play a fundamental role in areas such as cartography, computer graphics
and computer vision.

In this work, we consider the problem of \emph{terrain simplification},
also known as \emph{terrain approximation}.
We refer to~\cite{hg-survey,talton-survey,cms-comparison}
for extensive surveys on the topic and focus on the concepts important for
our work.
Generically, given a terrain $T$ over a domain, it asks for a ``simpler''
triangulation $T'$ of the same domain which constitutes a ``good'' approximation
of the original terrain $T$. To make this setup concrete, we must specify the meaning of ``good'' and ``simple''. For the former,
a standard choice (which we also adopt throughout the paper) is to impose
a maximal pointwise vertical distance between original and approximate terrain.
Formally, for $\eps>0$, we want that $\|T-T'\|_\infty <\eps$,
where $T$ and $T'$ are interpreted as scalar functions on a common domain.
Another possible choice is to impose a quality criterion on the triangular
mesh of $T'$ (e.g., being a Delaunay mesh).

For defining ``simplicity'', a default choice is ask for a terrain with
fewer vertices, but natural formulations
for this optimization problem are NP-hard (see the related work section below).
\emph{Topological simplification} is an alternative way to define simplicity,
where the goal is to reduce the number of \emph{critical vertices} (minima,
maxima, saddles) of the terrain. Unlike in the previous case, the
optimization problem is tractable:
the algorithm by Bauer, Lange, and Wardetzky~\cite{blw-optimal}
(called BLW algorithm from now on)
produces an $L_\infty$-close terrain
with the minimal number of critical points.
As discussed in~\cite{blw-optimal}, such a ``topologically clean'' approximation is useful in many applications:
for instance, in order to identify drainage basins, it is desirable to
remove spurious minima that lead to a too fine fragmentation of the terrain.
However, the drawback of the BLW algorithm
is that realizing the topological simplification as a terrain requires
a barycentric subdivision of the original triangulation,
which increases the size of the terrain by a factor of~$6$ and thus
results in a severe performance penalty of subsequent steps.

The BLW algorithm makes use of the popular concept
of~\emph{persistence diagrams}~\cite{edelsbrunner2010computational}.
Such a diagram partitions the critical points into pairs $(p,q)$
such that $p$ and $q$ can both be removed from the terrain using
a pointwise perturbation of \emph{persistence} $\Delta(p,q)$,
which is the height difference
of the two critical points. Bauer et al.~show that it is possible to
remove \emph{all} critical pairs with $\Delta(p,q)\leq 2\eps$
with a single $\eps$-perturbation. Since critical point pairs
with $\Delta(p,q)>2\eps$ cannot be cancelled in this way due
to the stability of persistent homology~\cite{ceh-stability},
this proves the optimality.

\paragraph{Contributions}
We investigate the following question: under what conditions does the persistence diagram
of the terrain remain the same when
(1) flipping an edge, or (2) removing a (regular) vertex and triangulating its link?
For both cases, we propose sufficient and necessary conditions that can be checked
locally (Section~\ref{sec-theory}). While the edge flip condition can be checked
in constant time, the vertex removal condition can be implemented
with a $O(d^3)$ algorithm with dynamic programming, where $d$ is the degree
of the removed vertex.
 Both tests can be easily combined with testing for pointwise closeness,
that is, whether an edge flip/a vertex removal yields a $L_\infty$-close
terrain with the same persistence.

Using the above test, we suggest a simple post-processing procedure
for the output terrain of the BLW algorithm (Section~\ref{sec-algorithm}): traverse the
regular vertices and greedily remove vertices without changing the persistence
diagram, always maintaining an $L_\infty$-close terrain, until no further
vertex can be removed. Note that maintaining the persistence diagram
implies that the number of critical points remains the same. Hence,
the result of our post-processing still achieves the minimal number
of critical points, but is smaller in size.

We experimentally evaluate our method (Section~\ref{sec-experiments}).
For instance, on a terrain with $100$K
vertices, the BLW algorithm yields a topologically clean terrain
with about $600$K vertices. Our post-processing yields a topologically equivalent
terrain that only consists of $11$K vertices.
Hence, our method addresses the major drawback of the BLW algorithm
of returning a too large terrain.

The above example can be computed in about 2 minutes on a workstation.
Achieving this running time requires several algorithmic
ideas. One of them is a heuristic improvement of the BLW algorithm
to avoid computing the entire barycentric subdivision.
Our implementation
is based on \textsc{Cgal}'s arrangement package~\cite{cgal-arrangements}
and uses exact number types for numerical computations. The code is available in a public repository\footnote{https://bitbucket.org/mkerber/terrain\_simplification.}.

\paragraph{Motivation and further related work}
Bajaj and Schikore~\cite{bs-topology} outline a reduction method that is
similar to ours: they propose to remove vertices from the terrain
and re-triangulate, such that the resulting terrain is $L_\infty$-close
and all vertices retain their ``criticality type'' (i.e., being regular,
a minimum, saddle or maximum). We show that preserving the persistence
diagram is equivalent to this condition. While they outline a method
for checking whether a re-triangulation meets their conditions,
they do not describe how to find such a triangulation efficiently,
and they do not report on experimental results.

Our approach is motivated by the area of
hierarchical models of terrains, where one aims
for a multi-resolution representation of a terrain that
should reflect the essential properties of the terrain
at various levels of detail.
To create such representation, it is useful to act on geometric and topological properties of the terrain
separately.
Several algorithms for modifying the topology
while maintaining the triangulation of the terrain
have been proposed \cite{Brem04, Wein09, Gyul11, Gunt12, Comi13-ismm, Felle14};
our approach can be seen as acting in the other direction, maintaining topology
and simplifying the triangulation.
Formerly proposed
methods described in \cite{iuricich2017hierarchical, Dey201833} aim for a similar goal, but they focus on the
restricted case of edge contractions and on how such operators affect discrete Morse gradient vector fields.



The BLW algorithm is an improvement over previous simplification
algorithms by Attali et al.~\cite{aghlm-persistence} and Edelsbrunner et al.~\cite{emp-persistence} which are not guaranteed to remove
all critical points of persistence $\leq 2\eps$.
As with the BLW algorithm,
these algorithms require a barycentric subdivision, so our method can be
combined with these algorithms as well.

Our work investigates the practical aspects of terrain
reduction and does not discuss the optimality of our (greedy) removal
strategy. Without topological constraint, the problem
is already difficult;
precisely,
the problem of finding the smallest terrain that is $L_\infty$-close to an input terrain is NP-hard~\cite{as-surface}
but approximation algorithms~\cite{ad-efficient}
and many heuristic approaches with weak or no guarantees on the size of the approximation (e.g., \cite{eck1995multiresolution}) exist.


\section{Background notions}\label{sec-background}

\paragraph{Triangulated terrains}\label{sub:complexes}
A {\em terrain} $\terrain$ is a triangulation of a compact polygonal region in $\R^2$, possibly with internal vertices, endowed with an injective scalar function $\terrainheight: V \rightarrow \R$ called {\em height function} defined on its set of vertices $V$; the injectivity of $\terrainheight$ is just assumed for simplicity in the write-up but our implementation does not require it.
Using barycentric coordinates, we can piecewise-linearly extend the height function $\terrainheight$ to the entire domain of $\terrain$. Based on that, a terrain $\terrain$ is always associated with a triangulated surface in $\R^3$ obtained as the graph of such an extended height function (see Figure \ref{fig:terrain} for an example). Whenever this will not cause any ambiguity, we will make no distinction between considering a terrain $\terrain$ as a triangulation in $\R^2$ endowed with a height function $\terrainheight$ or as the corresponding surface.
Given two terrains $\terrain$ and $\altterrain$ defined on the same polygonal region $D$, we define $\|\terrain-\altterrain\|_\infty:=\max_{p\in D} |\terrainheight(p)-\altterrainheight(p)|$.

\begin{figure}[!htb]
	\centering
        \includegraphics[width=.4\linewidth]{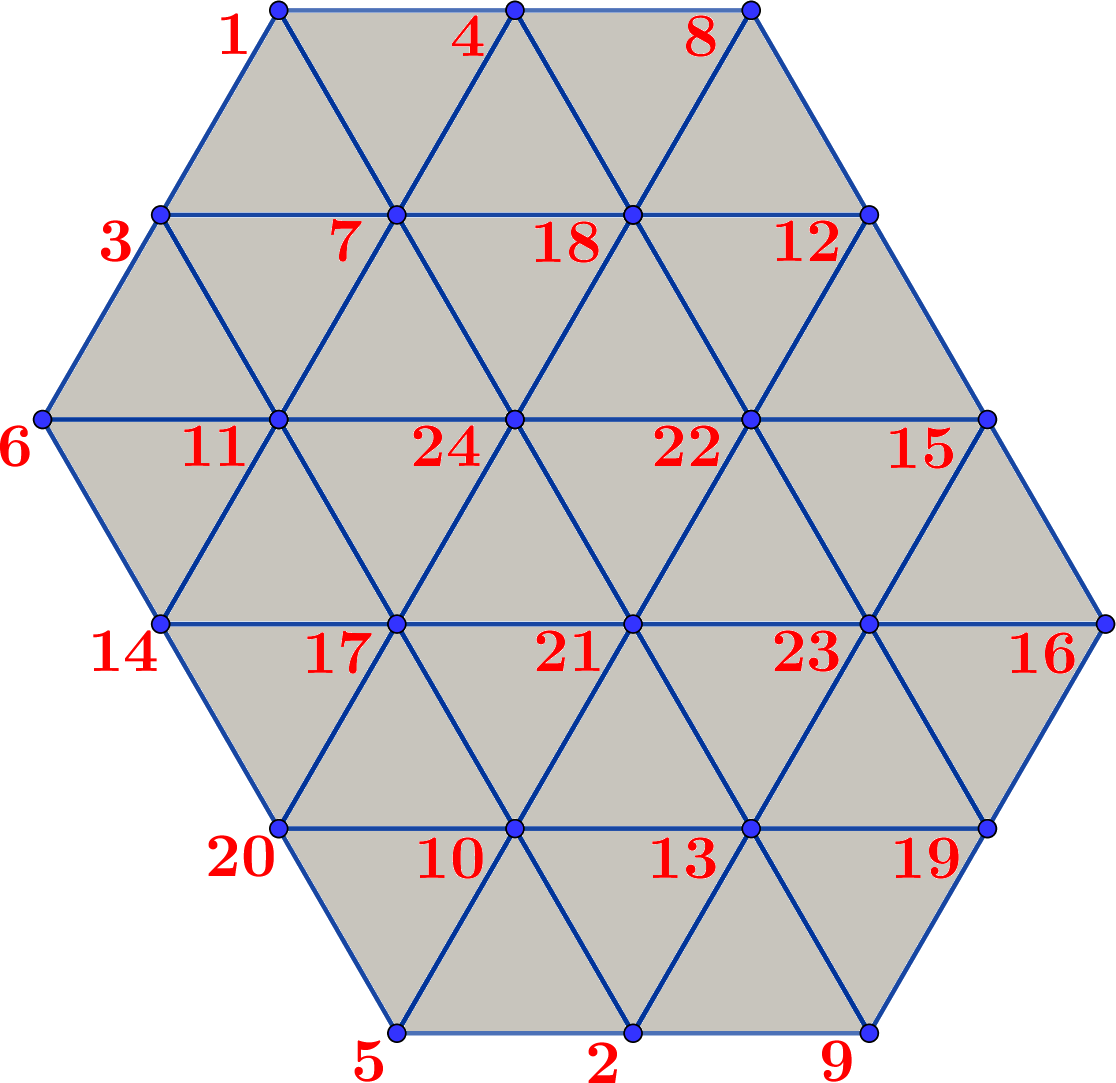}
        \hspace{0.8cm}
	\includegraphics[width=.45\linewidth]{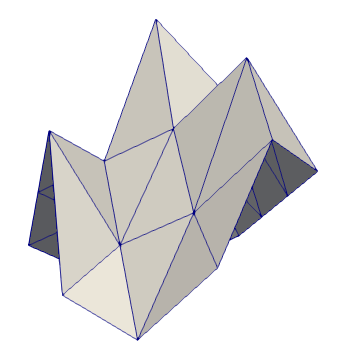}

	\caption{A terrain $\terrain$ represented as a triangulation of a polygonal region in $\R^2$ endowed with an injective scalar function $\terrainheight: V \rightarrow \R$ (left) and as a surface in $\R^3$ (right).}
  \label{fig:terrain}
\end{figure}

A vertex $v$ is called \emph{interior} if it is not on the boundary of the
triangulated domain.
For an interior vertex $v$, the {\em link} $\link(v)$ of $v$ consists of all the vertices adjacent to $v$ as well as all the edges $e=ab$ of $\terrain$ such that $abv$ is a triangle in $\terrain$.
The {\em lower link} $\link^-(v)$ of a vertex $v$ of $\terrain$ is the collection of vertices $u$ and edges $e=ab$ in $\link(v)$ such that $\terrainheight(u)\leq \terrainheight(v)$ and $\max\{\terrainheight(a), \terrainheight(b)\} \leq \terrainheight(v)$, respectively.
Analogously, the {\em upper link} $\link^+(v)$ is the collection of vertices and edges in $\link(v)$ satisfying the above equations in which $\leq$ is replaced with $\geq$.
Let us call an interior vertex $v \in \terrain$ {\em regular} if both $\link^-(v)$ and $\link^+(v)$ are non-empty and connected.
Otherwise, $v$ will be called {\em critical}.
As an example, in Figure \ref{fig:terrain}(a), vertex 10 is regular and vertices 17 and 24 are critical.

\paragraph{Persistent homology of a terrain}\label{sub:homology}
Given a value $\alpha \in \R$,
we write $L_\alpha$ for all points in the domain whose height is at most $\alpha$. Then $L_\alpha\subseteq L_\beta$, and we call $(L_\alpha)_{\alpha\in\R}$
the \emph{piecewise-linear (PL) filtration} of the terrain.

\begin{figure}[!htb]
	\centering
	\begin{tabular}{ccc}
		\includegraphics[width=.43\linewidth]{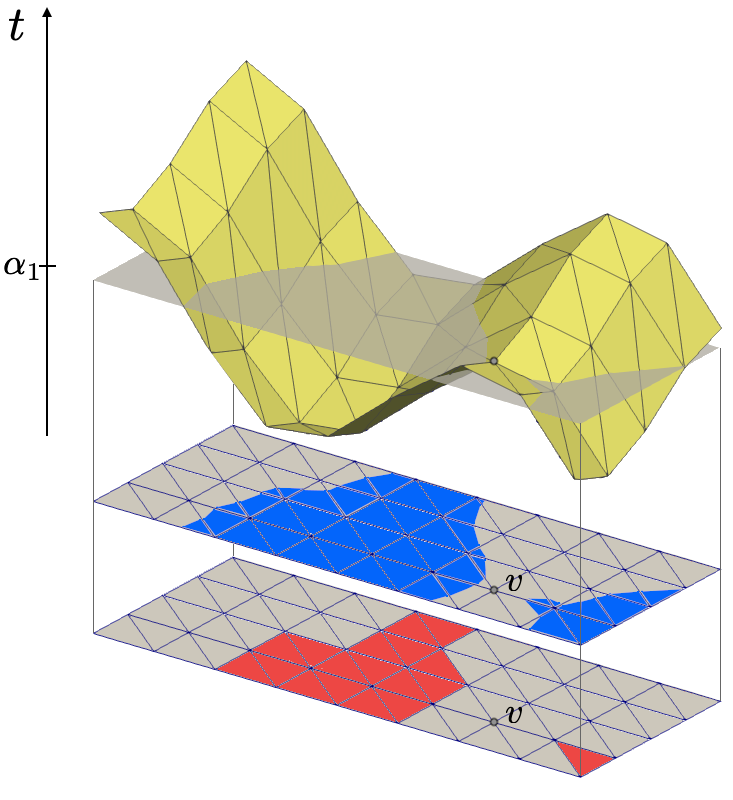} & &
		\includegraphics[width=.43\linewidth]{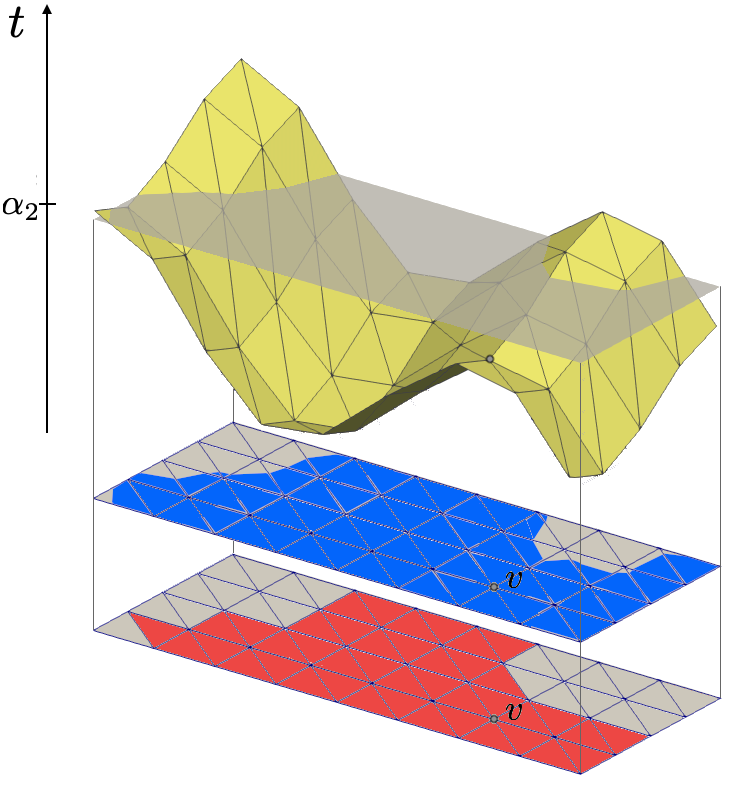} \\

		(a) & & (b)\\
	\end{tabular}
	\caption{The sublevel sets of the PL filtration (in blue) and of the simplex-wise filtration (in red) of a terrain $\terrain$ obtained for values $\alpha_1$ (a) and $\alpha_2$ (b).
	Independently from the chosen filtration, persistent homology records that the number of connected components of the blue (or red) domain decreases by one passing from $\alpha_1$ to $\alpha_2$.
}
\label{fig:persistence}
\end{figure}

\emph{Persistent homology}~\cite{elz-topological,oudot-book,edelsbrunner2010computational} enables us to study the topological changes occurring during a filtration. In our concrete case, persistence tracks the evolution of connected
components and holes in the terrain while the height function is increasing
(see Figure \ref{fig:persistence}).
It can be proven that, modulo a suitable handling of the boundary vertices,
there is a one-to-one correspondence between the critical points of $\terrain$ and the homological changes in its filtration~\cite{blw-optimal}.
The information gathered by persistent homology is summarized in a combinatorial
structure called the \emph{persistence diagram} (or equivalently, the \emph{barcode}).

It will be convenient to work with a different filtration in our setting.
For each point $p$ of a terrain $\terrain$,
we define $\sigma_p$ as the lowest-dimensional cell (vertex, edge, triangle)
that contains $p$. Assuming that $\sigma_p$ is spanned by boundary vertices
$v_1,\ldots,v_i$ with $i\in\{1,2,3\}$, we set $s_p:=\max\{\terrainheight(v_1),\ldots,\terrainheight(v_i)\}$,
and $S_\alpha:=\{p\mid s_p\leq\alpha\}$, and call $(S_\alpha)_{\alpha\in\R}$
the \emph{simplex-wise filtration} of the terrain.
Note that $s_p$ is not continuous in $p$, so the sets $S_\alpha$
change discontinuously at vertex values. Nevertheless, we still have that
$S_\alpha\subseteq S_\beta$, so the persistent homology of the simplex-wise
filtration is well-defined.
See Figure~\ref{fig:persistence} for an illustration.

It is well-known (e.g., ~\cite{blw-optimal}) that the PL filtration
and the simplex-wise filtration yield identical persistence diagrams.
This follows from the fact that, for every $\alpha$,
$S_\alpha\subseteq L_\alpha$ and, moreover, there is a deformation retraction~\cite{hatcher} from $L_\alpha$ to $S_\alpha$. This implies that the inclusion
map induces an isomorphism of homology groups. We will use the same
condition for the equality of persistence diagrams several times in this text.

\section{Persistence-preserving operations on terrains}
\label{sec-theory}

\paragraph{Persistence-preserving edge flip}

Given a terrain $\terrain$ and an interior edge $e=ab$, let $abc$ and $abd$ be the two triangles incident to $ab$ such that the quadrilateral $acbd$ is convex.
In that case, we call an \emph{edge flip} the operation of removing $ab$ and its two adjacent triangles from the terrain and replacing them
with the edge $cd$ and the triangles $acd$ and $bcd$, resulting in a new terrain $\altterrain$ (with coincides with $\terrain$ outside the quadrilateral $acbd$).

For two vertices $x,y$ of $\terrain$, we define the real interval
\[I_{xy}:=[\min\{\terrainheight(x),\terrainheight(y)\},\max\{\terrainheight(x),\terrainheight(y)\}].\]
We call the edge $ab$ \emph{topologically flippable} if $I_{ab}\cap I_{cd}\neq\emptyset$.
See Figure~\ref{fig:edge_flip_proof} for an example.

\begin{figure}[!htb]
  \centering
  \includegraphics[width=.47\textwidth]{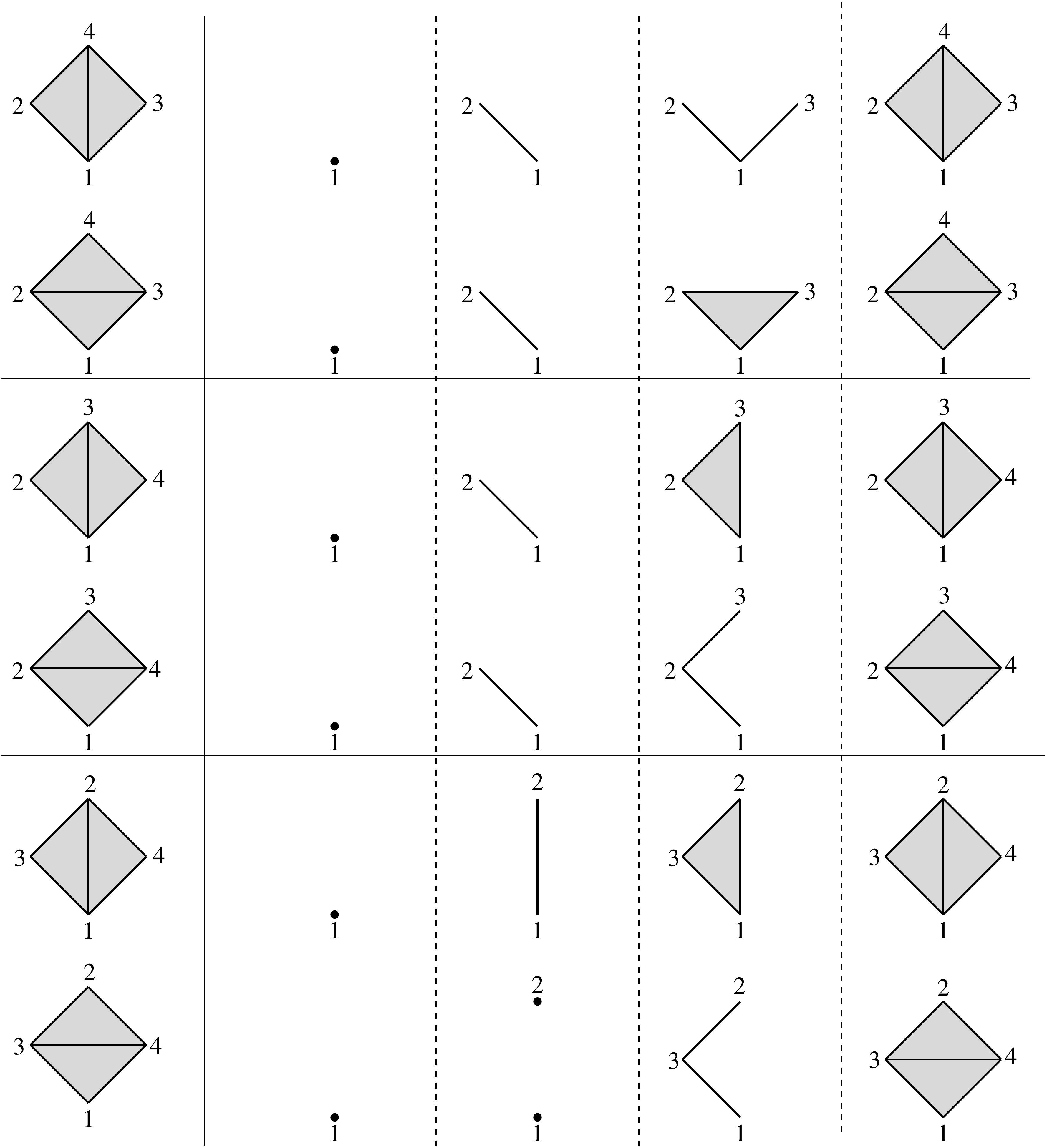}
  \caption{Left: in the first four rows, the diagonal of the quadrilateral is topologically flippable,
in the last two rows not. On the right, the simplex-wise filtration of the triangulation is displayed.}
  \label{fig:edge_flip_proof}
\end{figure}

\begin{proposition}\label{pro:flip}
Let $\terrain$ be a terrain with an edge $ab$ as above, and let $\altterrain$ denote the terrain after flipping the edge $ab$.
The edge $ab$ is topologically flippable if and only if the terrains $\terrain$ and $\altterrain$ have the same persistence diagrams.
\end{proposition}
\begin{proof}
We can, without loss of generality, assume that the height values of $a,b,c,d$ in $\terrain$ are $\{1,2,3,4\}$ and that the height of $a$ is equal to $1$.
Then, there only exist three combinatorial configurations for the quadrilateral $acbd$, which are depicted in the leftmost column of rows $1$, $3$, and $5$ of Figure~\ref{fig:edge_flip_proof}. The edge $ab$ is the vertical edge in each picture. Rows $2$, $4$, and $6$ show the corresponding flipped situation.

Using the fact that the persistence diagram of a terrain is determined by its simplex-wise filtration, it is enough to compare
the simplex-wise filtrations at the critical values $1$, $2$, $3$, and $4$. The right-hand columns of Figure~\ref{fig:edge_flip_proof} show the sublevel sets of the filtrations at these values
inside the quadrilateral $acbd$ (outside the quadrilateral, both terrains coincide).

Comparing row $1$ and $2$, we observe that the filtration of $\terrain$ is included in the filtration of $\altterrain$ for each critical value (even though that inclusion does not
respect the simplicial structure in the rightmost column).
We observe that the inclusion from $\terrain$ into $\altterrain$
is a deformation retract at value $3$, and trivially also at any other position
because the filtrations are equal everywhere else.
It follows that the two filtrations have the same persistence diagram.
The same argument can be applied to row $3$ and $4$, switching the roles of $\terrain$ and $\altterrain$.
This proves the ``if'' part of the statement.

Comparing row $5$ and $6$, we observe for critical value $2$, the sublevel sets of $\terrain$ and $\altterrain$ differ by exactly one edge.
By the Euler-Poincar{\'e} formula~\cite{edelsbrunner2010computational}, the homology groups of the two sublevel sets differ,
which implies that the persistence diagrams are not equal. This proves the ``only if'' part.
\end{proof}

The proof also reveals that if an edge is not topologically flippable, the
two persistence diagrams differ by at most the distance of $I_{ab}$ and
$I_{cd}$ in bottleneck distance (see~\cite{edelsbrunner2010computational}
for the definition). This might be of interest in variants of simplification
where a small change of persistence is acceptable.

\paragraph{Persistence-preserving vertex removal}

Given a terrain $\terrain$ and an interior vertex $v$ of $\terrain$ with exactly $d$ incident edges, the \emph{vertex removal} of $v$ is the operator which removes $v$ from $\terrain$ together with its $d$ incident edges and triangles,
and triangulates the link of $v$ using a set of $d-3$ diagonals, resulting in the terrain $\altterrain$.
See Figure~\ref{fig:removal_illu} for an illustration.

\begin{figure}[!htb]
\centering
\includegraphics[width=.47\textwidth]{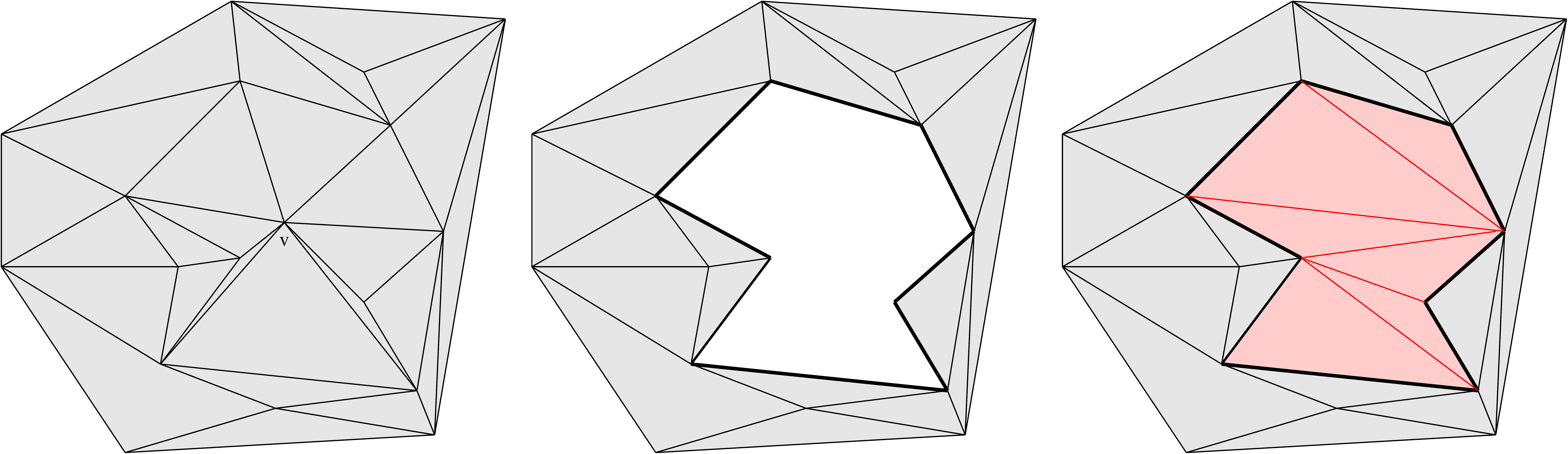}
\caption{Left: a triangulation with a vertex $v$ of degree $8$. Middle: the vertex and all its incident edges and triangles removed. The link of $v$ is
drawn thicker. Right: a possible re-triangulation of the link  with $5$ diagonals.}
\label{fig:removal_illu}
\end{figure}

Again, we are interested in circumstances under which the persistence diagrams of $\terrain$ and $\altterrain$ coincide. A necessary condition is that $v$ is regular,
as one can readily check. If the link of $v$ is a triangle, it is also sufficient.

\begin{figure}[!htb]
\centering
\includegraphics[width=.47\textwidth]{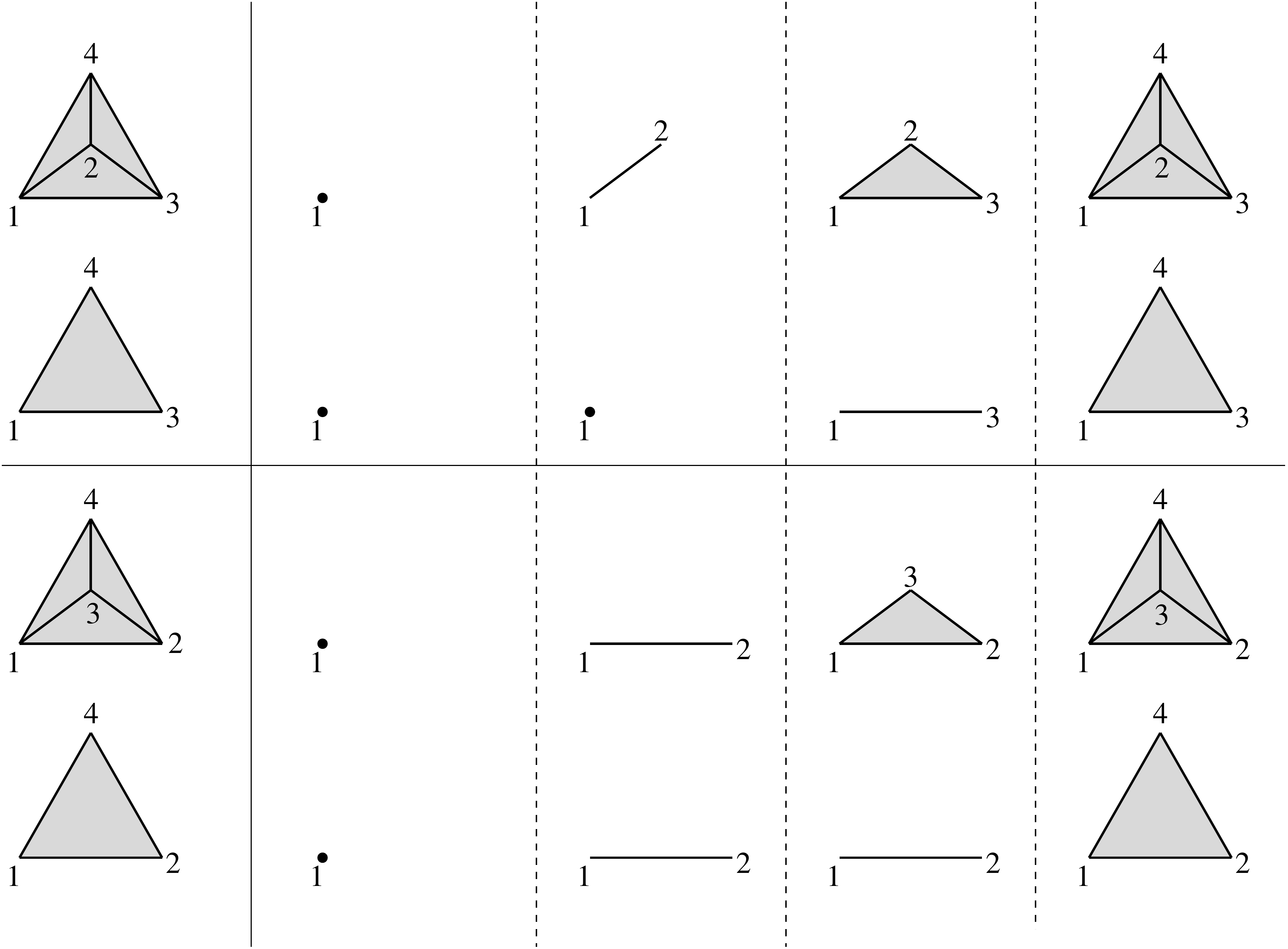}
\caption{Row $1$ and $3$: the two possible configurations of a regular vertex with a triangular link, and their simplex-wise filtrations.
Row $2$ and $4$: the same with the vertex removed.}
\label{fig:triangular_proof}
\end{figure}

\begin{proposition}\label{pro:triangle}
Let $v$ be a regular vertex of $\terrain$~with three incident edges. Then, the vertex removal of $v$ yields a terrain $\altterrain$ with the same persistence diagram as $\terrain$.
\end{proposition}
\begin{proof}
Note that because the link of $v$ is a triangle, no diagonals are necessary for the re-triangulation. Denoting the three adjacent vertices of $v$ as $a$, $b$, and $c$,
we can assume that $v,a,b,c$ have heights $\{1,2,3,4\}$ and since $v$ is regular, its height is either $2$ or $3$.
Figure~\ref{fig:triangular_proof} depicts both possible situations. In both cases, it can be observed, as in the proof of Proposition~\ref{pro:flip},
that the persistence diagrams are the same.
\end{proof}

For arbitrary vertex removals, we have generally several choices for triangulating the link. To characterize
which triangulations are persistence-preserving, we use the following concept.

\begin{definition}\label{def:pers-aware}
Let $v$ be a vertex of $\terrain$ with two adjacent vertices $a$, $b$ such that $\terrainheight(a)<\terrainheight(b)$.
We call the edge $ab$ \emph{persistence-aware} if one of the following conditions holds:
\begin{itemize}
  \item $\terrainheight(a)<\terrainheight(v)<\terrainheight(b)$;
  \item $\terrainheight(b)<\terrainheight(v)$ and there is a path on the link of $v$ from $a$ to $b$ with maximal height $\terrainheight(b)$;
  \item $\terrainheight(v)<\terrainheight(a)$ and there is a path on the link of $v$ from $a$ to $b$ with minimal height $\terrainheight(a)$.
\end{itemize}
\end{definition}

\begin{figure}[!htb]
    \centering
    \includegraphics[width=.23\textwidth]{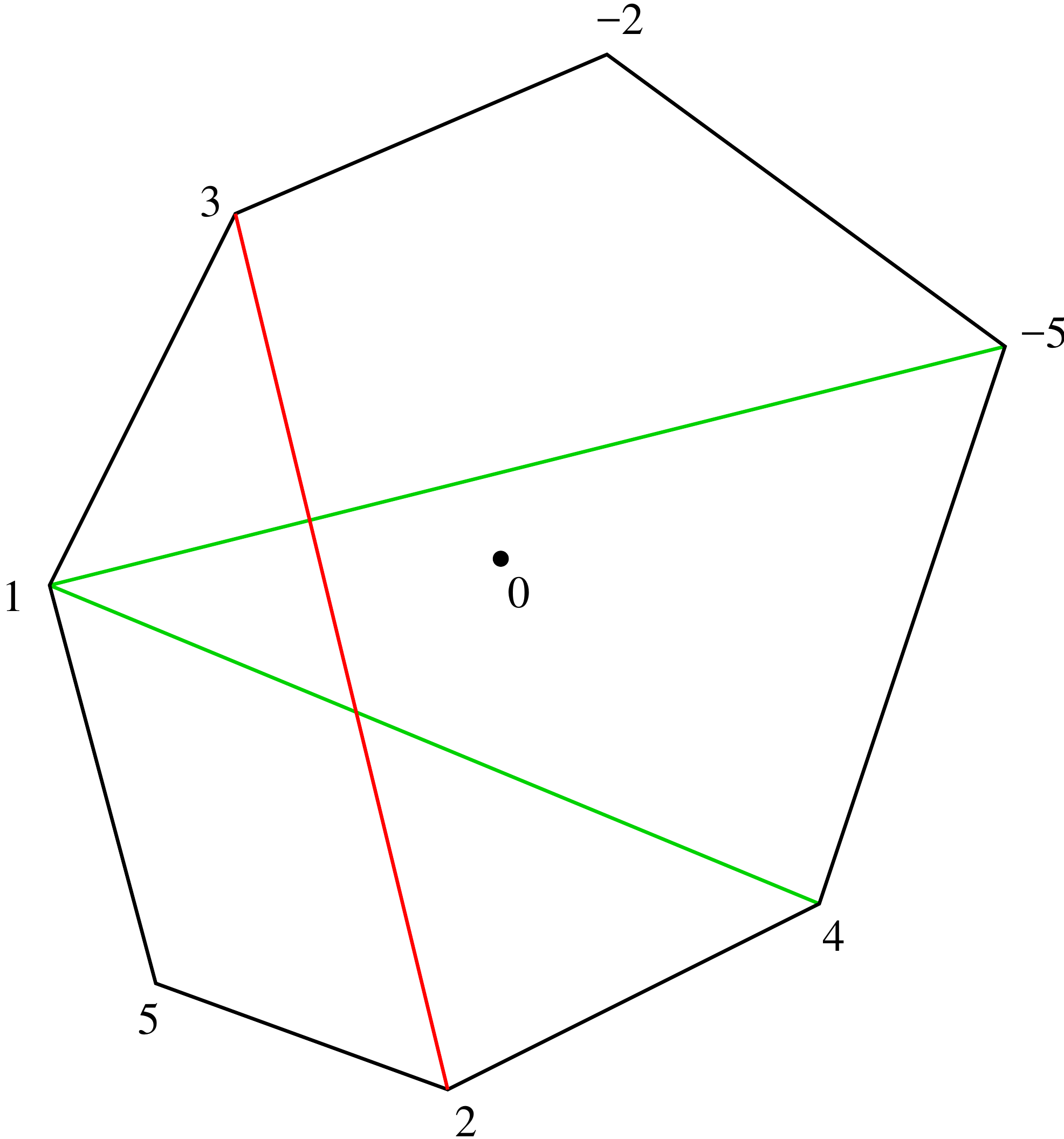}
    \caption{The link of a vertex with height $0$ is displayed. We identify the vertices in the link with their height. The edge $(-5)1$ is persistence-aware because
one height is positive and one negative (first case of Definition~\ref{def:pers-aware}). The edge $14$ is also persistence-aware because on the path $1\rightarrow5\rightarrow2\rightarrow4$, $1$ has minimal height
(second case of Definition~\ref{def:pers-aware}). The edge $23$ is not persistence-aware because $2$ is not minimal on either path from $2$ to $3$.}
\label{fig:pers-aware}
\end{figure}

An example is depicted in Figure \ref{fig:pers-aware}.

\begin{theorem}\label{thm:removal}
Let $v$ be a regular interior vertex of $\terrain$ of degree $d$. Let $\altterrain$ denote the terrain obtained by removing $v$ and re-triangulating its link
with $d-3$ diagonals. If all diagonals are persistence-aware, $\terrain$ and $\altterrain$ have the same persistence diagram.
\end{theorem}

The proof idea is illustrated in Figure~\ref{fig:remove_thm_illu}.
We point out that the converse is also true, and the condition
is also equivalent to the statement that all vertices in the link
maintain their criticality status. See \ref{app:extended_proof}
for the proofs of these statements.

\begin{figure*}[!htb]
    \centering
    \includegraphics[width=.8\textwidth]{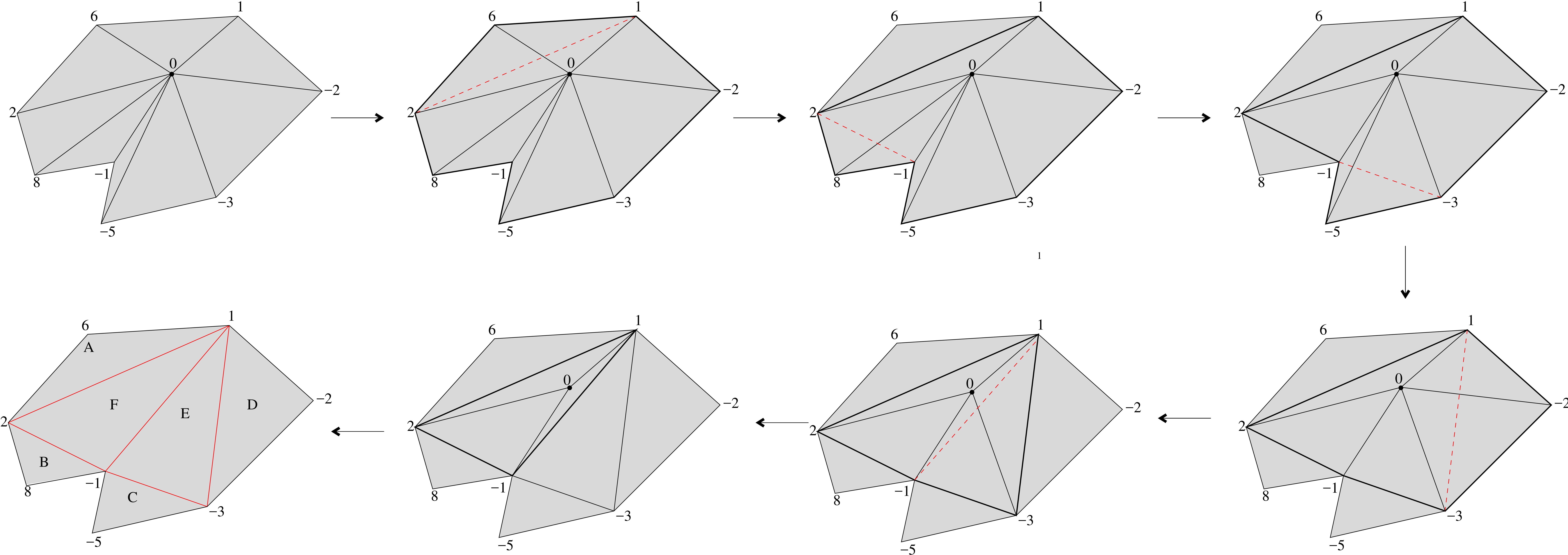}
    \caption{Illustration of the proof strategy of Theorem~\ref{thm:removal}. The upper left shows the link of a regular vertex to be removed. The lower left shows a triangulation of its link where all edges are persistence-aware. The remaining steps show the transformation used in the proof: in every step, the link of $v$ decreases by one vertex through a flipping of a topologically flippable edge.
After the last step, the link of $v$ is a triangle, which can be removed with Proposition~\ref{pro:triangle}.}
\label{fig:remove_thm_illu}
\end{figure*}

\begin{proof}
Without loss of generality, we assume that $\terrainheight(v)=0$.
We prove the statement by induction on $d$. Note that $d\geq 3$ because $v$ is not on the boundary. The case $d=3$ follows from Proposition~\ref{pro:triangle}.
For $d>3$, we write $\linktri$ for the triangulation of the link of $v$ using the $d-3$ persistence-aware diagonals. An \emph{ear} in a triangulation of a polygon is
a triangle $acb$ where both $ac$ and $bc$ are polygon edges, and only $ab$ is a diagonal. It is a well-known fact~\cite{de2000computational} that
every triangulation of a polygon with $d\geq 4$ vertices has at least two ears.

We call an ear $acb$ \emph{irregular} if $\terrainheight(a)$ and $\terrainheight(b)$ have the same sign and $\terrainheight(c)$ has the opposite sign,
and \emph{regular} otherwise. Note the the presence of an irregular ear implies that either the lower or the upper link of $v$ consists
of the vertex $c$ only (because $v$ is regular). Consequently, $\linktri$ can have at most one irregular ear, and hence, must have at least one
regular ear.

Now, let $acb$ denote a regular ear of $\linktri$. The idea is to flip the edge $cv$ of $\terrain$ to obtain the edge $ab$, which results in
$v$ losing one vertex in its link, and to use induction. Let us first make the simplifying assumption that $v$ is outside of the triangle $acb$,
hence $acbv$ forms a convex quadrilateral. We argue that $I_{ab}\cap I_{cv}\neq\emptyset$, which implies that flipping $cv$ does
not change the persistence diagram by Proposition~\ref{pro:flip}:
if the signs of $\terrainheight(a)$ and $\terrainheight(b)$ differ, this is clear because $\terrainheight(v)=0$.
If the signs of $\terrainheight(a)$ and $\terrainheight(b)$ are the same, the sign of $\terrainheight(c)$ is the same, because the ear is regular.
We assume that all signs are positive; the negative case is analogous. Because the edge $ab$ is persistence-aware, there is a path
on the link where all height values are at least $\min\{\terrainheight(a),\terrainheight(b)\}$. The path from $a$ to $b$ that does not pass
through $c$ passes through the lower link of $v$ and hence contains at least one vertex with negative height.
So, the only path that can have the property asserted by persistence-awareness is $a\rightarrow c\rightarrow b$. This implies $I_{ab}\cap I_{cv}\neq\emptyset$.

Let $\altterrain$ denote the terrain obtained from $\terrain$ by flipping $cv$. As we just argued, $\terrain$ and $\altterrain$
have the same persistence diagram. In $\altterrain$, $v$ has degree $d-1$. Moreover, the $d-4$ remaining diagonals (apart from $ab$)
triangulate the link of $v$ in $\altterrain$, and it can be readily checked that they are all persistence-aware, because
they were persistence-aware in $\terrain$. Hence, by induction, we can re-triangulate the link without changing the persistence diagram,
and the claim follows.

It remains to deal with the case that $v$ is inside the (regular) ear.
In this situation, the quadrilateral $acbv$ is not convex, and the edge flip
introducing $ab$ yields to crossings.
A first idea might be to just relocate $v$ inside
the polygon (the position of $v$ inside the polygon
has no effect on the persistence diagram).
It is possible, however, that the relocated $v$ might not see
all boundary vertices anymore (see Figure~\ref{fig:visibility_issue} (left)).

Instead, we can resolve this issue
by not insisting that the edges of a terrain are straight lines.
We omit a formal description of the argument for the sake of simplicity
and outline the proof idea.
A \emph{topological triangulation}
is a planar embedding of a graph where all bounded faces
are bounded by three edges. In a topological triangulation,
an edge flip can always be realized,
possibly by bending edge, as displayed in Figure~\ref{fig:visibility_issue} (right).
We can also define the simplex-wise filtration of
topological triangulations because the filtration only depends on the
elevation of the vertices and the combinatorial structure of the triangulation;
in particular, the filtration, and the resulting persistence diagram
are independent of the actual planar embedding.
Also, the definition of links and persistence-aware diagonals,
and all results from this section extend to the topological case.
In that way, we obtain a topological triangulation with $d-3$ diagonals
and the same persistence diagram as the initial one. But, since we
obtain the diagonals by cutting off ears, a simple inductive argument
shows that these diagonals do not cross when embedded
as straight-line segments. This allows us to ``straighten'' the final triangulation and finishes the proof.
\end{proof}

\begin{figure}[!htb]
    \centering
    \includegraphics[width=.2\textwidth]{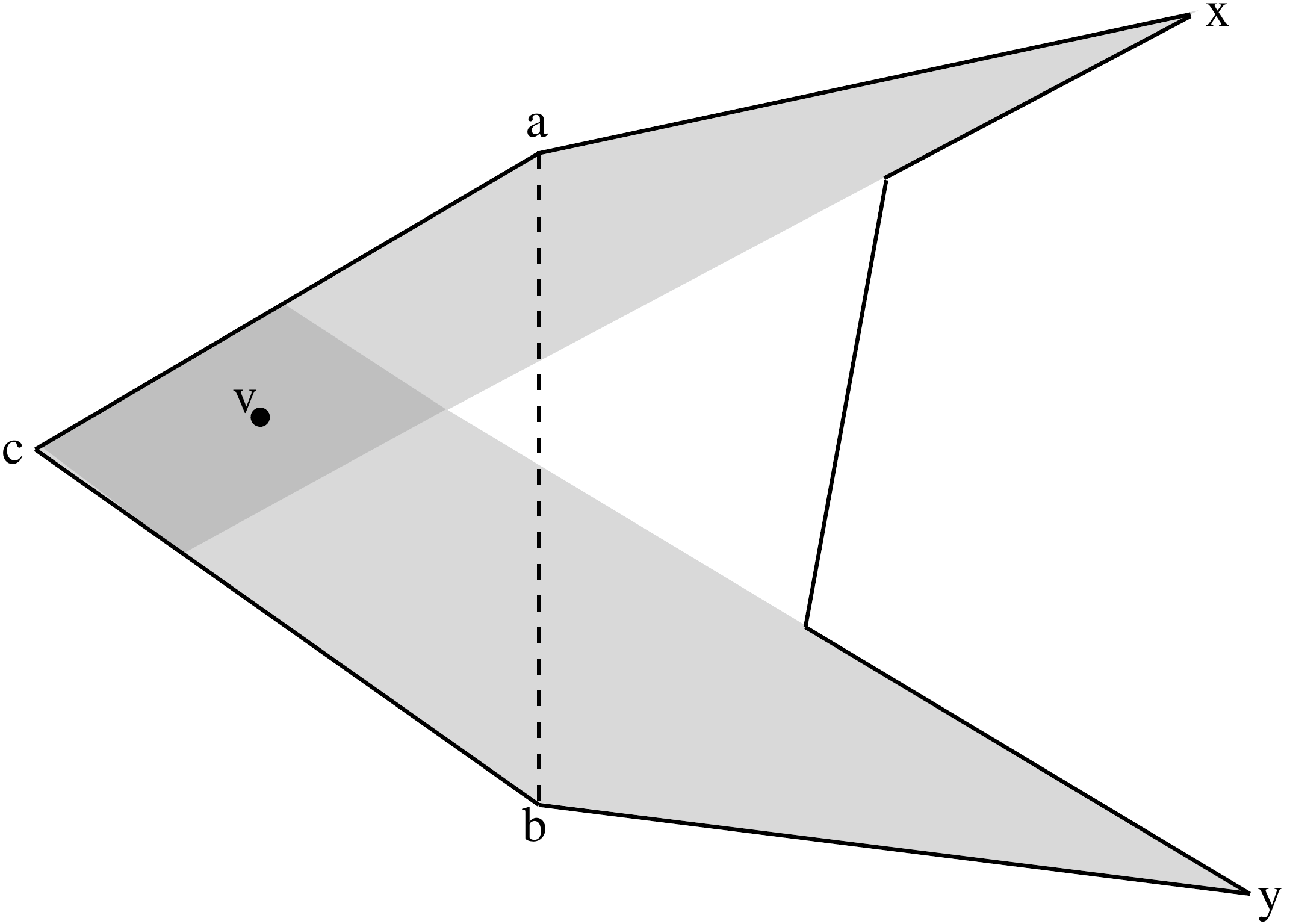}
    \includegraphics[width=.2\textwidth]{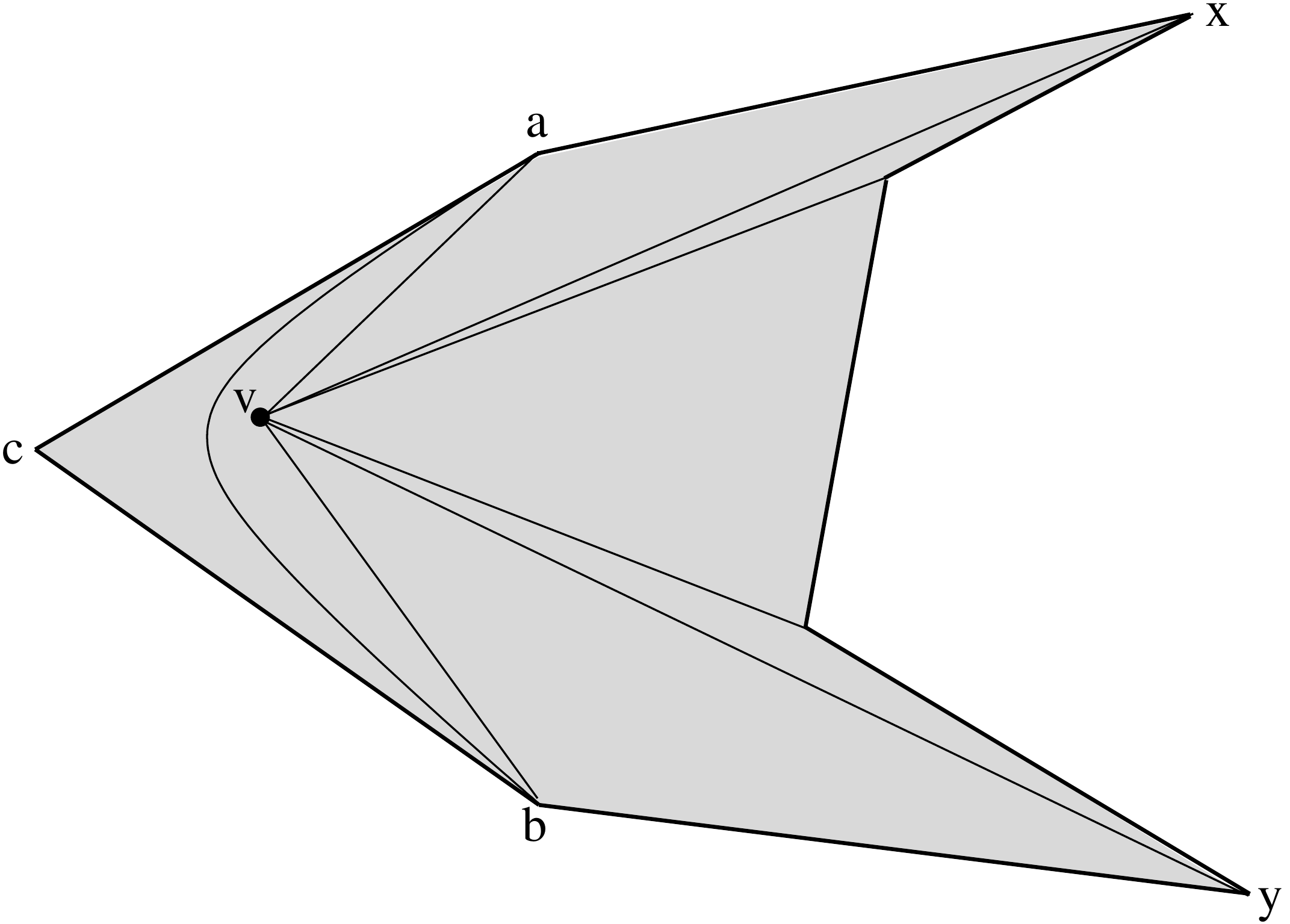}
    \caption{Left: the vertex $v$ cannot be moved out of the triangle $acb$ such that it sees both $x$ and $y$,
as illustrated by the shaded regions. Right: we can flip $cv$ to $ab$ by
bending $ab$.}
\label{fig:visibility_issue}
\end{figure}

\section{Algorithm}
\label{sec-algorithm}
\paragraph{Reduction method}
We give an algorithm for the following problem: given $\eps>0$, a \emph{base terrain} $\baseterrain$ and a terrain $\terrain$ with $\|\baseterrain-\terrain\|_\infty\leq\eps$, compute a terrain $\simpterrain$ with fewer vertices than $\terrain$, such that
$\|\baseterrain-\simpterrain\|_\infty\leq\eps$ and $\terrain$ and $\simpterrain$ have the same persistence diagram.
Note that $\baseterrain=\terrain$ is allowed, in which case one obtains
a $L_\infty$-close approximation of $\terrain$ with the same topological
information. In our application, however, $\baseterrain$ will be the input
terrain, and $\terrain$ will be the result of the BLW algorithm,
arising from a subdivision of $\baseterrain$.

The algorithm idea is to start with $\simpterrain\gets\terrain$
and incrementally removing vertices from $\simpterrain$,
maintaining the condition on the $L_\infty$-distance and the persistence
diagram as specified above.
For that, we proceed as follows: we maintain a set of candidate vertices
(initially set to all vertices of $\terrain$). While this set is non-empty,
we choose a candidate vertex $v$ uniformly at random,
remove it from the candidate set
and check whether it can be removed and its link
can be triangulated maintaining the condition.
If yes, we remove $v$ from $\simpterrain$ and
add the edges of the link triangulation. We re-insert the vertices of the link
to the candidate set, if not already contained, and start over.
This finishes the description of the algorithm.

\paragraph{Testing edges and triangles}
The major primitive of the above algorithm is to find a triangulation of the
link of a vertex that maintains the algorithm's condition (or to output that
there is none). We call a vertex \emph{removable} if such a triangulation
exists. To characterize removable vertices, we need the following definitions.
Let $\simpterrainheight$ denote the height function of $\simpterrain$
and $\baseterrainheight$ the height function of $\baseterrain$.
For two vertices $u$, $v$ of $\simpterrain$ (not necessarily connected
via an edge), we call the line segment $uv$
\emph{$L_\infty$-aware} if for any point $x=\lambda u+(1-\lambda)v$
on the line segment $uv$, we have that $|(\lambda\simpterrainheight(u)+(1-\lambda)\simpterrainheight(v)) - \baseterrainheight(x)|\leq\eps$.
In other words, if $uv$ was an edge of the terrain $\simpterrain$,
the $L_\infty$-distance of $\simpterrain$ and $\baseterrain$
is at most $\eps$ along the edge. Similarly, for three vertices $u$, $v$, $w$
of $\simpterrain$, we call the triangle $uvw$ \emph{$L_\infty$-aware}
if for any point $x=\lambda u+\mu v+(1-\lambda-\mu)w$ of the triangle,
we have that $|(\lambda\simpterrainheight(u)+\mu\simpterrainheight(v)+(1-\lambda-\mu)\simpterrainheight(w)) - \baseterrainheight(x)|\leq\eps$.

\begin{lemma}\label{lem:removable}
A vertex $v$ is removable if there is a triangulation of its link such that
all diagonals of the triangulation are persistence-aware and $L_\infty$-aware,
and all triangles of the triangulation are $L_\infty$-aware.
\end{lemma}
\begin{proof}
Let $\simpterrain_1$ denote the terrain before the removal of $v$,
and $\simpterrain_2$ the terrain when removing $v$ and inserting the triangulation. Since both diagonals and triangles of the link triangulation are
$L_\infty$-aware, $\simpterrain_2$ is $\eps$-close to $\baseterrain$
in the area enclosed by the link, and hence on the entire domain,
because it coincides with $\simpterrain_1$
everywhere else. The persistence diagrams of $\simpterrain_1$
and $\simpterrain_2$ are equal by Theorem~\ref{thm:removal}
because all diagonals are persistence-aware.
\end{proof}

Note that checking whether a diagonal is persistence-aware can be done
by simply traversing the link once, yielding a $O(d)$ primitive
with $d$ the size of the link. The $L_\infty$-awareness tests are based
on the following results, which reduces the locations where the height
difference is maximized to a finite set.

\begin{lemma}
\label{lem:overlay}
The $L_\infty$-distance of $\simpterrain$ to $\baseterrain$ along an edge $uv$
is maximized at $u$, at $v$, at a crossing of $uv$ with an edge of $\baseterrain$, or at an intersection of $uv$ with a vertex of $\baseterrain$.

The $L_\infty$-distance of $\simpterrain$ to $\baseterrain$ in a triangle $uvw$
is maximized at a boundary edge or vertex of $uvw$,
or at a vertex of $\baseterrain$ inside the triangle.
\end{lemma}

\begin{figure}[!htb]
    \centering
    \includegraphics[width=.23\textwidth]{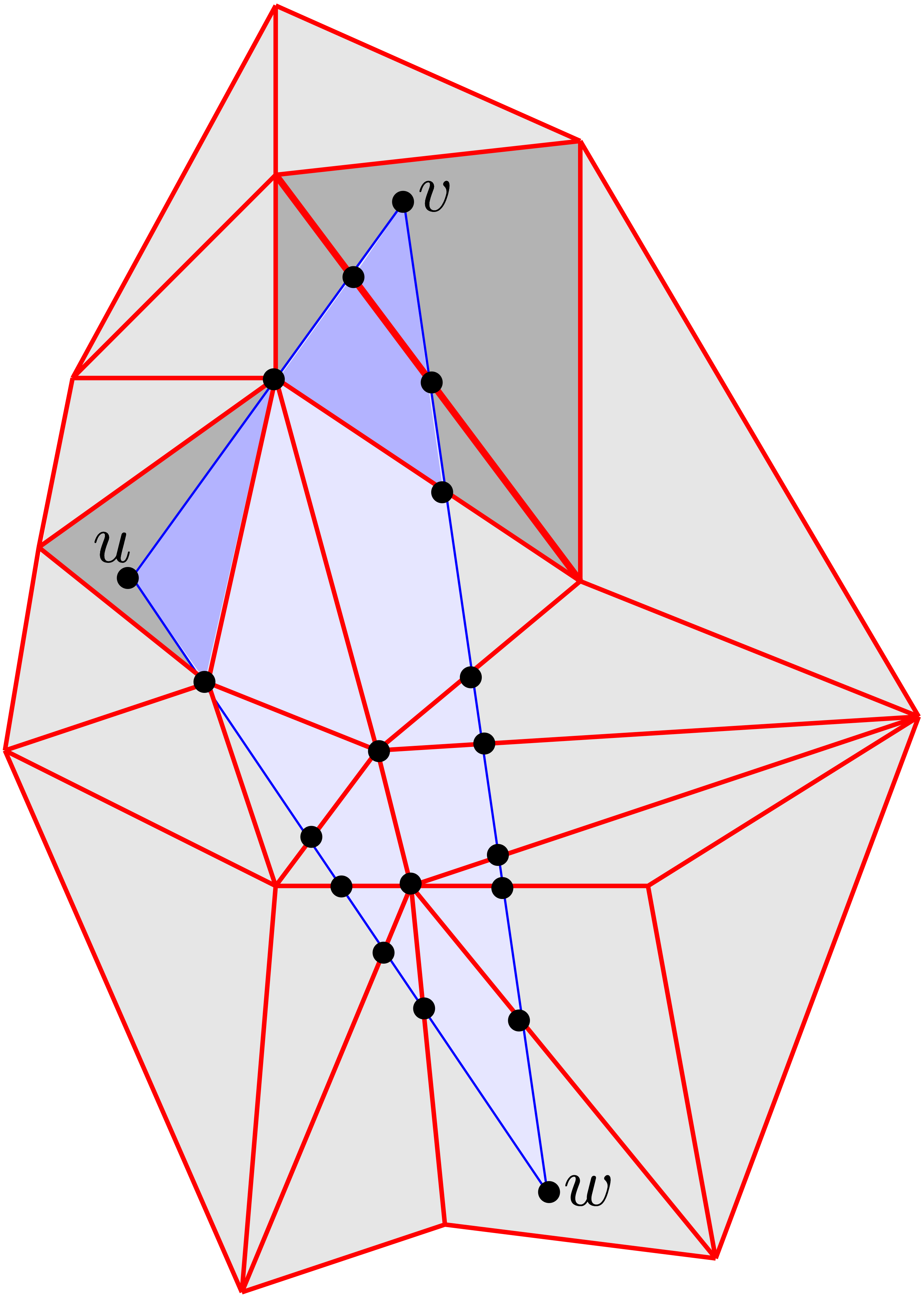}
    \caption{The maximal difference between the terrain $\baseterrain$ (in red)
and the triangle $uvw$ of $\simpterrain$ (in blue) is attained at one
of the marked points (in black). For each boundary edge, the maximum
is attained at one of the marked points along this edge. The zone
of the edge $uv$ is also marked (darkly shaded).}
\label{fig:overlay}
\end{figure}

We refer to Figure~\ref{fig:overlay} for an illustration
and to \ref{app:algo_details} for the (simple) proof.

Hence, for checking $L_\infty$-awareness of edges, it suffices
to compute the \emph{zone} of $uv$ in $\baseterrain$ (i.e., the set of vertices, edges,
and faces of $\baseterrain$ traversed by $uv$) and check the height difference
at most one point per element in the zone.
The running time is dominated by computing the zone of $uv$.

For checking the $L_\infty$-awareness of triangles, it suffices
to compute the $L_\infty$-awareness of its boundary edges and to additionally
check the height difference for all vertices of $\baseterrain$
within $uvw$. We obtain these vertices by a triangular range query
in $\baseterrain$, and the running time is dominated by the complexity
of this range query - see \ref{app:algo_details}.

\paragraph{Finding a triangulation}
The above primitives to check persistence- and $L_\infty$-awareness
allow us to investigate a fixed triangulation. However, the naive application
of Lemma~\ref{lem:removable} by just trying all triangulations is prohibitive
because of the exponential number of possible triangulations of the link.
We describe a simple dynamic programming approach.

Let $P$ denote the link polygon and let its vertices be denoted by $1,\ldots,d$
in counterclockwise order. Let $1\leq i<i+1<j\leq d$, and define $P[i,j]$
as the subpolygon spanned by vertices $i,i+1,\ldots,j$.
The idea is that any triangulation of $P[i,j]$ has one triangle $ikj$
incident to the edge $ij$, where $i<k<j$. To determine whether $P[i,j]$
has a valid triangulation with triangle $ikj$, it suffices to check
whether the edges $ik$ and $kj$ are persistence- and $L_\infty$-aware,
whether the triangle $ikj$ is $L_\infty$-aware, and whether the
subpolygons $P[i,k]$ and $P[k,j]$ admit a valid triangulation
(if $k=i+1$ or $j=k+1$, the corresponding test can be skipped). We simply
iterate through all $k=i+1,\ldots,j-1$ until a valid triangulation
is found or all values of $k$ have been tested. Applying this
procedure to $P=P[1,d]$ yields an algorithm to check whether a vertex
is removable (and computes a valid triangulation in the positive case).

The above procedure operates on subpolygons of the form $P[i,j]$ with
$1\leq i< j\leq d$ and we can thus store all intermediate answers
using quadratic space. We also store the persistence- and
the $L_\infty$-awareness of every edge when computed, using quadratic
space as well. Hence, the procedure requires a total of $O(d^2)$
such awareness tests on edges, and $O(d^3)$ $L_\infty$-awareness tests
on triangles in the worst case.\footnote{There is no reason to store
the results on triangles since every triangle is queried only once.}

\paragraph{Mesh improvement}
With similar ideas, we can also describe a greedy procedure to improve
the quality of the triangulation of a terrain, without
sacrificing its geometric and topological properties.
The input are two terrains $\baseterrain$
and $\terrain$, and the output is a terrain $\simpterrain$, with the same
relations as in the reduction phase, with the difference that
$\simpterrain$ has the same number of vertices as $\terrain$.

The idea is reminiscent of the well-known edge-flipping algorithm
for turning an arbitrary triangulation into a Delaunay triangulation
just by greedily flipping edges to increase the minimal angle locally~%
\cite{de2000computational}: we flip the edge $e$ to the edge $e'$
if $e'$ is persistence- and $L_\infty$-aware, the two incident triangles
to $e'$ are $L_\infty$-aware, and the flip improves the minimal angle.
We maintain a set of candidate edges (initially all edges of $\terrain$)
and flip a candidate if the described criteria are met. We repeat until
no candidate remains.

\paragraph{The BLW algorithm}
We use the topological simplification algorithm from~\cite{blw-optimal}
as described in detail in the PhD thesis of Bauer~\cite{bauer-thesis}.
We also use the symmetrizing approach~\cite[Section 6.2]{blw-optimal}
to reduce the $L_\infty$-perturbation per cell.

The output of the BLW algorithm is a simplex-wise function on the
input triangulation $\baseterrain$,
that is, we obtain a filtration value for each vertex,
edge, and face of the triangulation. In order to convert this into a terrain,
we pass to the barycentric subdivision $\terrain$ of $\baseterrain$,
and assign to each
vertex of $\terrain$ the filtration value of $\baseterrain$ as its height.
It is known
that the terrain obtained in this way has the same persistence diagram
as the initial function~\cite[Thm 7]{blw-optimal}. However, in order to
satisfy the $L_\infty$-constraint, we might have to place the vertex
very close to one of the vertices (see Figure~\ref{fig:barycenter_placement}),
which yields a poor quality mesh
and numerical instability (if an imprecise number type is used).
In general, such a placement might even be impossible
(see also~\cite[Section 4.3]{aghlm-persistence}),
but our solution works under the assumption that $\eps$ is chosen generically,
meaning that $2\eps$ is not equal to the height difference of any two input vertices.
We refer to \ref{app:algo_details} for details of our placement method.

\begin{figure}[!htb]
    \centering
    \includegraphics[width=.47\textwidth]{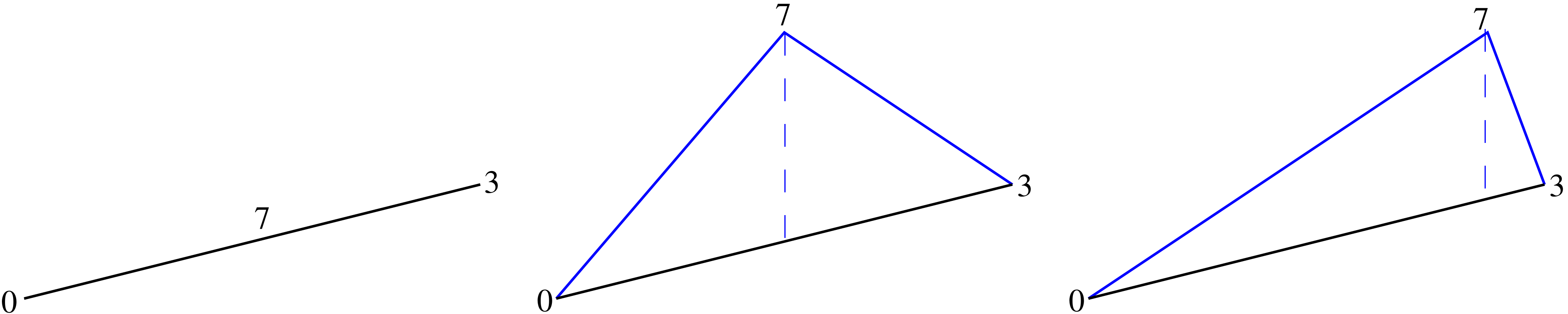}
    \caption{Illustration for the problem of barycentric subdivision. In this figure, the $y$-coordinate shows the height. Left: assume the BLW algorithm has returned the given height values for an edge and its
two endpoints (with $\eps=5$). Middle: splitting the edge in the middle
leads to a height difference of $6.5$. Right: placing the split point
sufficiently close to the right endpoint satisfies the $L_\infty$-constraint.}
\label{fig:barycenter_placement}
\end{figure}

\paragraph{Avoiding subdivisions}
Since our final goal is to remove as many vertices as possible, it is
beneficial to avoid subdividing edges and triangles if possible. We devise
a simple heuristic for that: if the BLW algorithm assigns to an edge the same
value as to one of its boundary vertices, we do not subdivide this edge.
Moreover, if the BLW algorithm assigns to a triangle the same value
as to one of its corner vertices and none of the three boundary edges
is subdivided, the triangle is not subdivided.
This procedure can create a situation that a triangle is subdivided,
but only some (or none) of its boundary edges are. In that case,
the vertex of the face is connected with the corner vertices
and the vertices of subdivided boundary edges
(see Figure~\ref{fig:sparse_subdivision}).
The correctness of the method comes from the observation that for each step
the filtrations of the fully subdivided and partially subdivided terrain
coincide.

\begin{figure}[!htb]
    \centering
    \includegraphics[width=.3\textwidth]{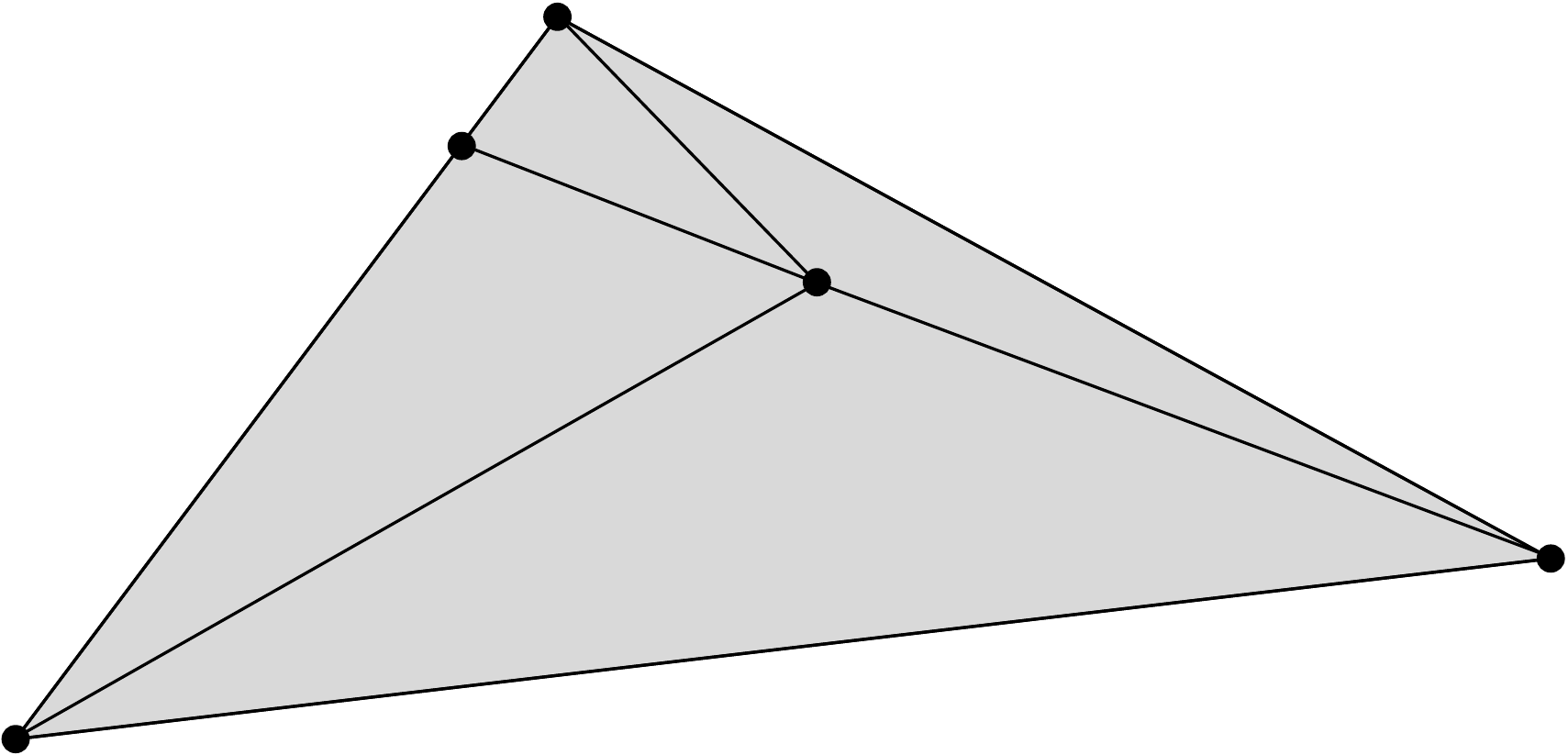}
    \caption{An example where the subdivision point of a triangle
is only connected to one subdivided edge, and the three boundary vertices.}
\label{fig:sparse_subdivision}
\end{figure}

\paragraph{The whole pipeline}
We summarize the entire algorithmic pipeline that we use. The input
is a terrain $\baseterrain$
and some $\eps>0$, such that $2\eps$ is not equal to the difference
of two heights. The output is another terrain $\simpterrain$
that is $\eps$-close to $\baseterrain$ in $L_\infty$-distance
and has the smallest number of critical points possible.

\begin{itemize}
\item Apply the BLW algorithm on $(\baseterrain,\eps)$ to get a
terrain $\terrain$ with the specified properties.
\item Apply the reduction method on $(\baseterrain,\terrain,\eps)$ to
get a smaller terrain $\simpterrain$ with the same properties.
\item Apply the mesh improvement algorithm on $(\baseterrain,\simpterrain,\eps)$
and return the output.
\end{itemize}

\section{Implementation and experiments}
\label{sec-experiments}

\paragraph{Implementation details}
We implemented the algorithm described in Section~\ref{sec-algorithm}
using C++. The program is available in a public repository\footnote{https://bitbucket.org/mkerber/terrain\_simplification.} and it consists of about 3000 lines of code.

Our implementation makes extensive use of the functionality provided
by the \textsc{Cgal} library (version 4.14)\footnote{CGAL, Computational Geometry Algorithms Library, https://www.cgal.org.}.
For representing terrains, we decided to use \textsc{Cgal}'s \emph{2D arrangement package}~\cite{cgal-arrangements}
instead of a triangulation data structure, because it provides more flexibility in the removal and re-triangulating
procedure. Moreover, it contains the basic functionality to traverse the zone of a line segment
in an arrangement, needed for checking the $L_\infty$-constraint for edges.
For triangular range queries,
we wrote our own algorithm~-- see \ref{app:algo_details}.

We point out that our geometric primitives cannot assume generic position of the vertices of the terrain.
Even if this was assumed for the input $\baseterrain$, subdividing an edge yields a triple of collinear
points by design. Moreover, as mentioned above, the subdivision approach might place vertices
very close to existing ones. To avoid instabilities due to these effects, our implementation uses
exact number types both for coordinates of points and height values
(provided by \textsc{Cgal}'s \texttt{Exact\_predicates\_exact\_construction\_kernel}~\cite{cgal-kernel}).

\paragraph{Examples}
We demonstrate the various steps of the proposed implementation with some examples.
We have considered datasets coming from two different sources: \cite{Steiermark} provides an altitude value for each point of a regular grid representing Styria, a state of Austria; \cite{Aim@shape} contains a collection of terrain models representing the area around some lakes in Italy (see Figure \ref{fig:example_input}).

\begin{figure}[!htb]
	\centering
  \resizebox{0.49\textwidth}{!}{
	\begin{tabular}{cc}
		\includegraphics[width=.6\linewidth]{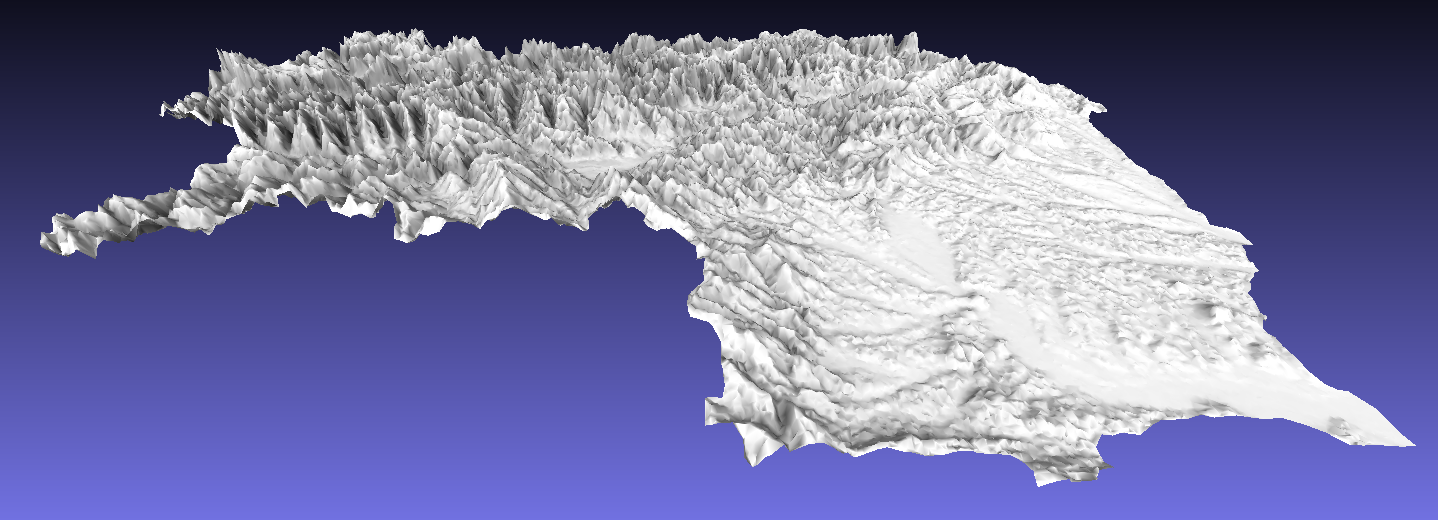} &
		\includegraphics[width=.318\linewidth]{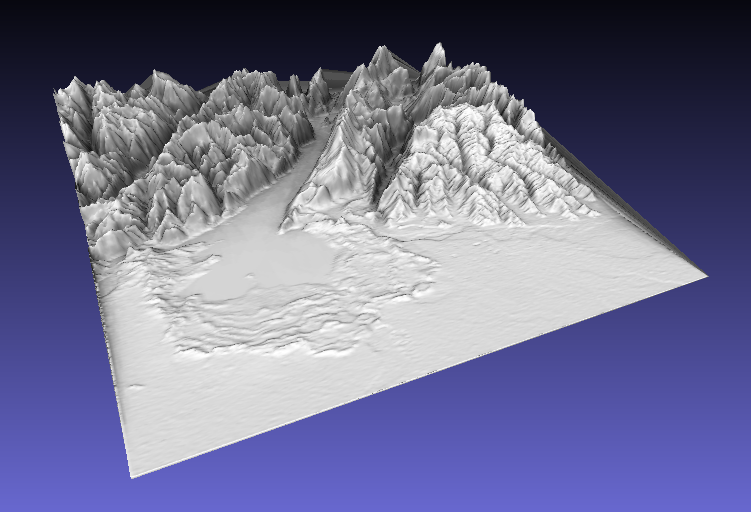} \\
		(a) & (b)\\
	\end{tabular}
  }
	\caption{Terrain representations of Styria (a) from \cite{Steiermark}, and of lake Garda (b) from \cite{Aim@shape}.}
  \label{fig:example_input}
\end{figure}

For both the classes of terrains, we have studied datasets of various sizes by considering all the provided points or a subsample of them.
We focus here on a representation of Styria of $100$K vertices (mentioned in the introduction), performing on it the proposed reduction strategy by choosing as $L_\infty$-constraint $\epsilon=100$ meters.
Despite the drastic reduction in size to $11$K vertices, the input and the output terrains look pretty similar (see Figure \ref{fig:example_IO-comparison}).

\begin{figure}[!htb]
	\centering
  \resizebox{0.49\textwidth}{!}{
	\begin{tabular}{cc}
		\includegraphics[width=.46\linewidth]{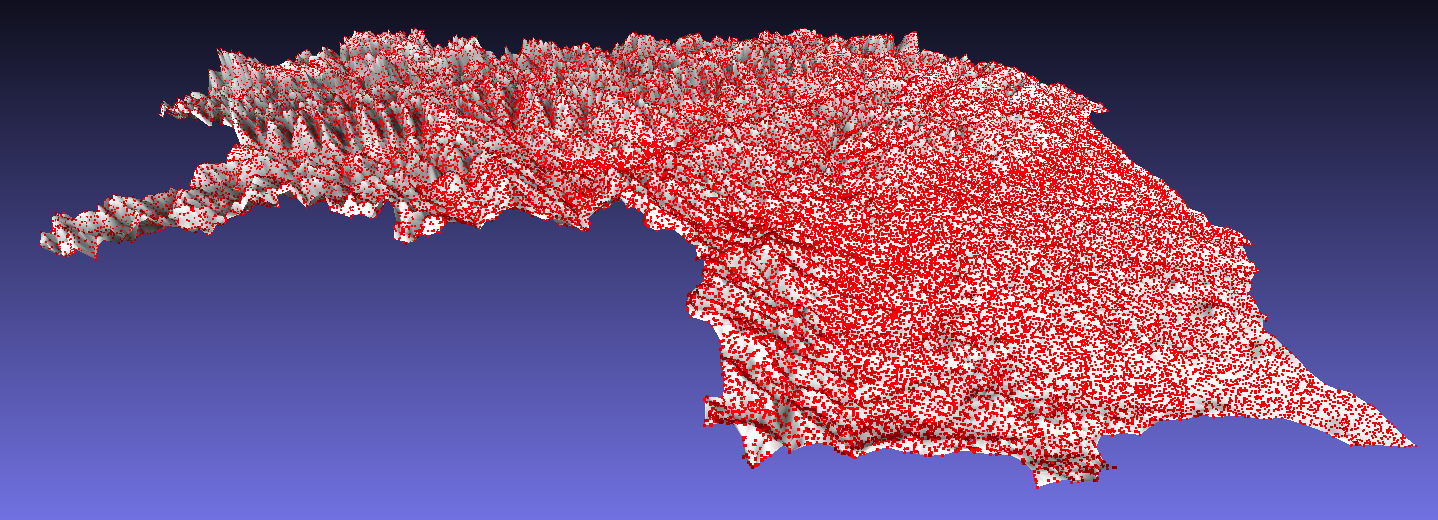} &
		\includegraphics[width=.46\linewidth]{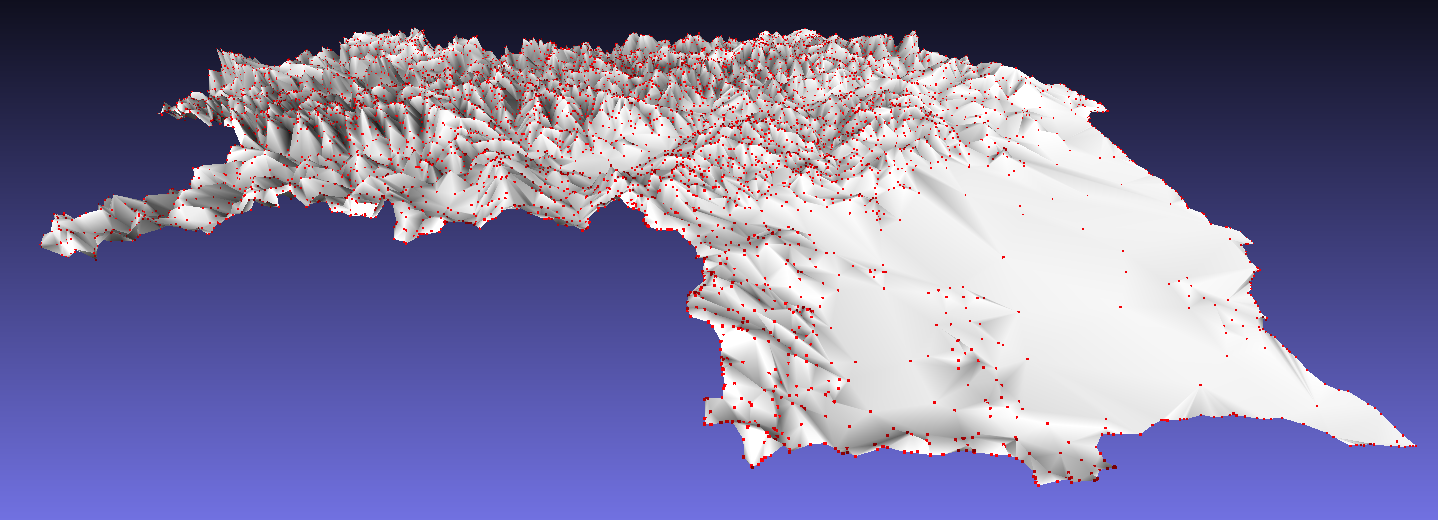} \\
		(a) & (b)\\
	\end{tabular}
  }
	\caption{The input terrain (a) and the output obtained by applying the proposed algorithm choosing $\epsilon=100$ meters (b). For both terrains, vertices are depicted in red.}
  \label{fig:example_IO-comparison}
\end{figure}

The first step of the proposed simplification strategy consists in performing the BLW algorithm. The topological simplification performed by BLW is based on a barycentric subdivision of the input terrain. As discussed in the previous section, this step can be improved by avoiding a certain number of subdivisions. Figure \ref{fig:example_Subdivisions} depicts cutouts of the input terrain, of the fully subdivided, and of the partially subdivided terrain. We observed that triangles involved in subdivisions in Figure \ref{fig:example_Subdivisions}(c) tend to form paths in the mesh surface rather than being randomly spread.

\begin{figure}[!htb]
	\centering
  \resizebox{0.49\textwidth}{!}{
	\begin{tabular}{ccc}
		\includegraphics[width=.29\linewidth]{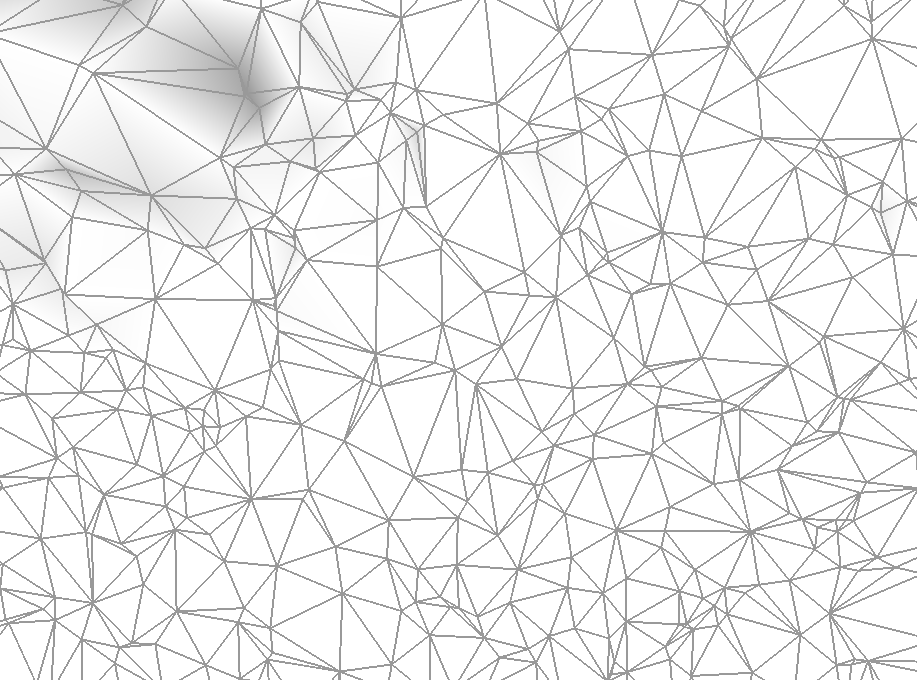} &
		\includegraphics[width=.29\linewidth]{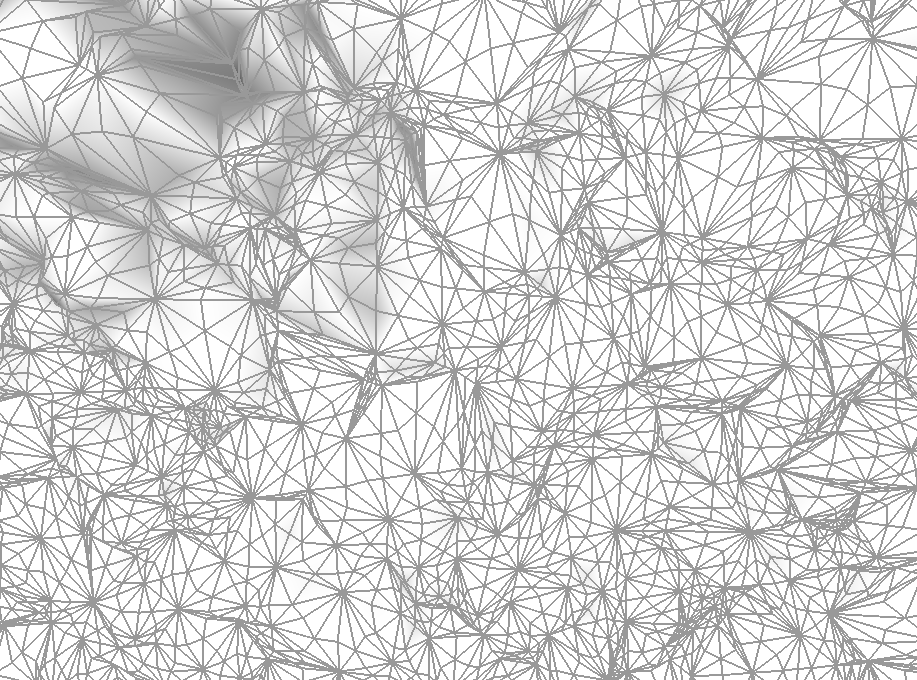} &
		\includegraphics[width=.29\linewidth]{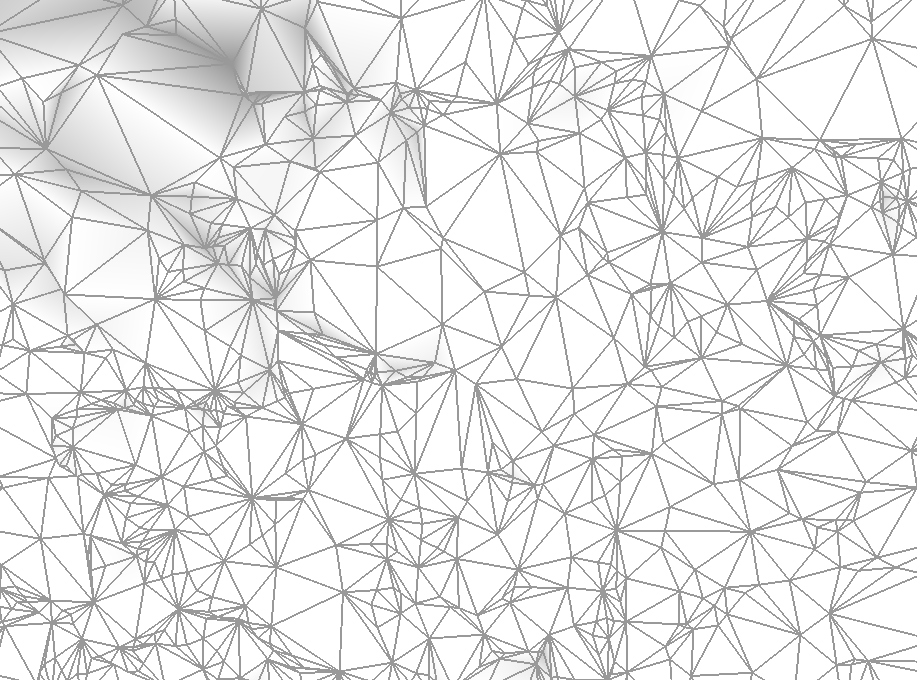}\\
		(a) & (b) & (c)\\
	\end{tabular}
  }
	\caption{A cutout of the terrain in input (a), of its barycentric subdivision (b), and of the terrain obtained by performing the proposed improvement of the BLW algorithm trying to avoid superfluous subdivisions (c).}
  \label{fig:example_Subdivisions}
\end{figure}

The last step of our algorithm improves the quality of the reduced terrain
by flipping edges to increase the minimal angle. In general, only a rather
small fraction of edges is flipped in this step (around 10\% in the above
example). Nevertheless, we can visually see local improvements, as depicted in Figure \ref{fig:example_Mesh-Improvement}.

\begin{figure}[!htb]
	\centering
  \resizebox{0.49\textwidth}{!}{
	\begin{tabular}{cc}
		\includegraphics[width=.46\linewidth]{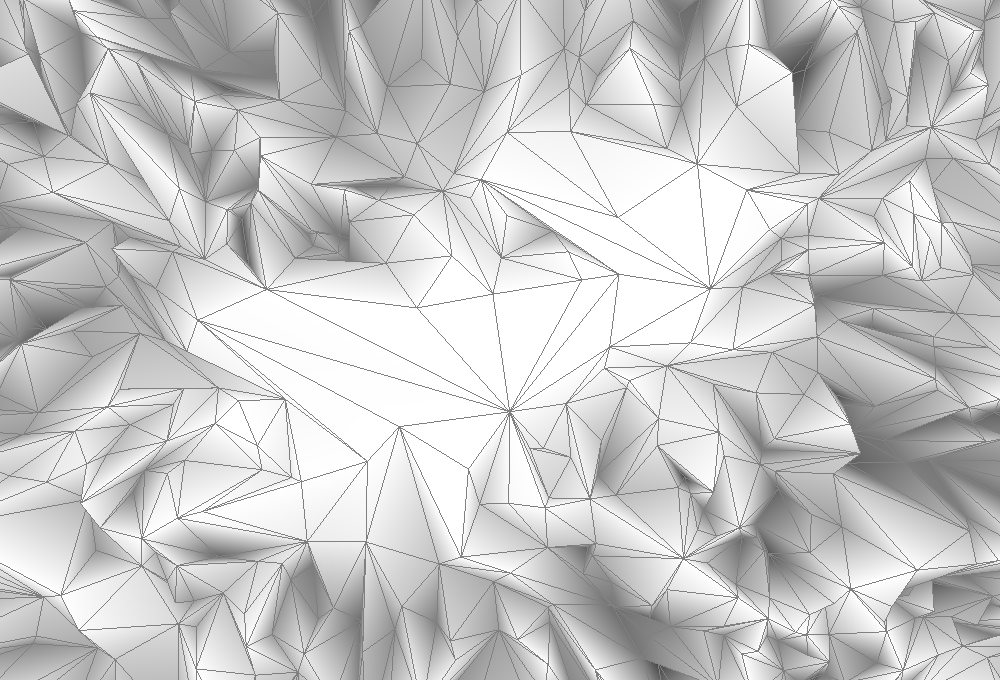} &
		\includegraphics[width=.46\linewidth]{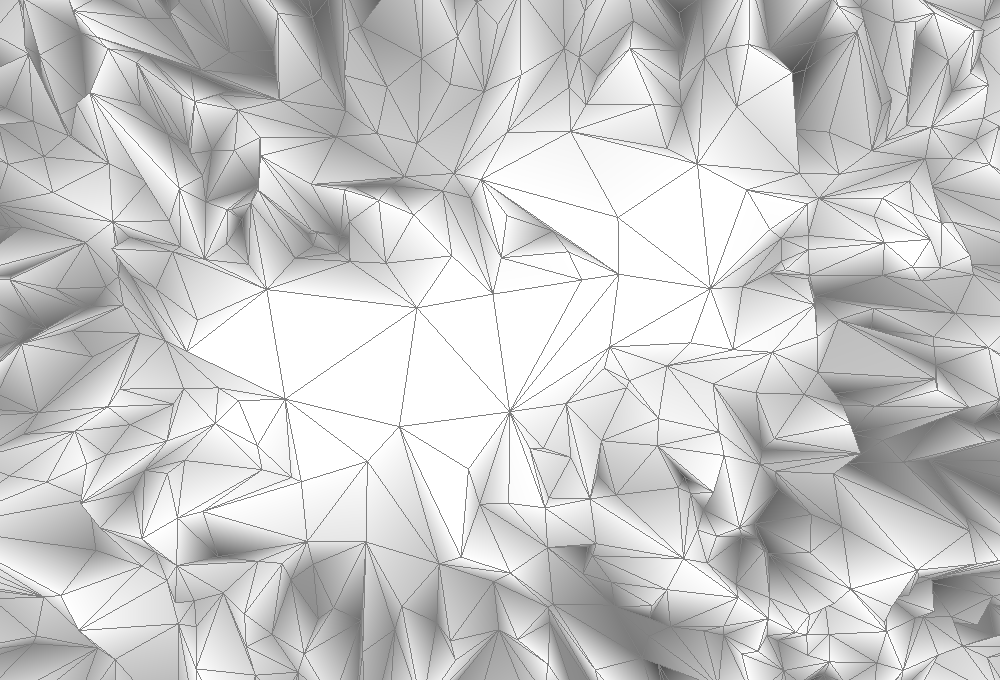} \\
		(a) & (b)\\
	\end{tabular}
  }
	\caption{A cutout of the output terrain before (a) and after (b) the mesh improvement obtained by flipping a selection of edges.}
  \label{fig:example_Mesh-Improvement}
\end{figure}

\paragraph{Performance}
We ran our implementation on a workstation with 6 CPU cores with 3.5 GHz
per core and 64 GB of RAM, running Ubuntu 16.04.5. Our tests do not
exploit the multi-core architecture and run on a single core.

In Table~\ref{tbl:runtimes}, we present the output size and the running time
of our procedure for subsamples of the Styria dataset. The sampled points were chosen uniformly at random. For each size, we ran the algorithm $5$ times and display the average running time and the maximal deviation from the average. For all runs, we picked $\epsilon=100$ meters.\footnote{An analogous table reporting the performances obtained for the lake datasets is presented in \ref{app:exp_results}.}

\begin{table}
\centering
\begin{tabular}{c|c|c|c|c}
I & S & C & O & T\\
\hline
10K & 17K & 407 & 4100.6 ($\pm 9.6$) & 8.90 ($\pm 0.18$)\\
20K & 33K & 543 & 6114.2 ($\pm 49.2$)  & 19.96 ($\pm 0.37$)\\
40K & 64K & 547 & 8654.2 ($\pm 27.2$)  & 38.96 ($\pm 0.36$)\\
80K & 124K & 597& 11223.2 ($\pm 74.8$)  & 81.19 ($\pm 1.79$)\\
160K & 239K & 611 & 13160.2 ($\pm 65.8$) & 168.54 ($\pm 3.74$)\\
320K & 455K & 637 & 15019.0 ($\pm 61.0$) & 342.40 ($\pm 6.12$)
\end{tabular}
\caption{Benchmark results for the Styria dataset. I is the size of the input (number of vertices), S is the size of the subdivision structure (the output of the BLW algorithm) using our improved subdivision strategy, C is the number of critical points of the BLW algorithm (and the output),  O is the output size of our method (number of vertices), and T is the running time of our method (in seconds). For O and T, we also show the largest deviation from the average.}
\label{tbl:runtimes}
\end{table}

First of all, we see that our method creates terrains of much smaller size
compared to the output of the BLW algorithm, and even much smaller than
the input terrain. Our sparser subdivision strategy generally decreases
the number of points by a factor of about $4$ compared
to the full barycentric subdivision whose size
is around six times as large as the input size.
This decrease leads to a a substantial saving of running time.
We also see that the output size and the number of critical points in the output terrain increase slowly compared to the input size. This is not surprising because a finer sampling of the landscape is unlikely to create persistent topological
features or drastic changes in the $L_\infty$-distance.
Moreover, we observe that the deviation from the average is small in all cases,
so the algorithm is rather stable regarding its randomized removal strategy.

The running time of the algorithm splits into about 9-10\% for creating
the initial arrangement data structure out of the input, 9-10\% for
computing the BLW simplification, 78-80\% for the reduction
and around 2\% for the final edge flips (for all problem sizes). Hence,
the reduction step is the bottleneck of the computation, but its performance
is less than an order of magnitude away from the other steps.

We also observe that the running time is slightly super-linear. To further investigate this,
we look at the average running time for an attempted vertex removal (successful or not)
at various stages of the algorithm. Figure~\ref{fig:removal_times} shows the average timings,
averaged over around 1000 attempts. We observe three things: first, the time per attempts
is almost the same at the beginning of the algorithm, independent of the input size. This
shows that the reduction step is indeed a local operation and only depends mildly (i.e., logarithmically)
on the size of the terrain. Second, we see that time per attempt gradually increases during the algorithm,
which is explained by the increased cost of computing a zone in the arrangement when the
edges get longer. Finally, we observe that the running time increases significantly towards
the end of the algorithm. This is due to the fact that the algorithm has more failed
attempts towards the end, and failed attempts tend to be more costly than successful ones
because more triangles need to be checked.

\begin{figure}[!htb]
\begin{center}
    \includegraphics[width=.47\textwidth]{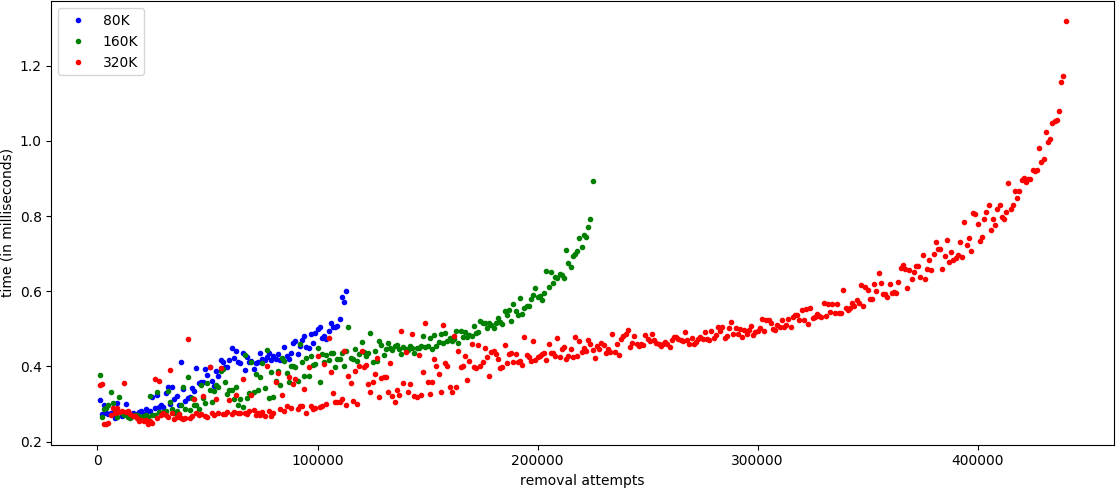}
\end{center}
    \caption{The average time per removal attempt (in milliseconds).}
    \label{fig:removal_times}
\end{figure}

\paragraph{Variants for vertex removal}
We experimented with two variants of our method: recall that in our original
version, when re-triangulating the link of a vertex $v$, we stop when
a valid triangulation is found. Instead, we can also keep checking
and return the valid triangulation with minimal $L_\infty$-distance
to the previous terrain. The idea is that a smaller distance
will be beneficial in subsequent steps to remove more vertices.
As we see in Table~\ref{tbl:heap_vs_greedy}, this strategy indeed yields
a modest decrease in the final number of vertices, in exchange for
a small slow-down of the method.

\begin{table}
\centering
\resizebox{0.49\textwidth}{!}{
\begin{tabular}{c||c|c||c|c||c|c}
 & \multicolumn{2}{c||}{Strategy 1} & \multicolumn{2}{c||}{Strategy 2} & \multicolumn{2}{c}{Strategy 3} \\
I & O & T & O & T & O & T\\
\hline
10K & 4094.6 & 9.00 & 4087.8 & 11.15 & 4063 & 47.73\\
20K & 6120.0 & 19.15 & 6095.0 & 25.00 & 6081 & 128.21\\
40K & 8648.2 & 39.41 & 8619.4 & 53.49 & 8623 & 365.91\\
80K & 11206.4 & 81.23 & 11128.6 & 115.62 & 11142 & 1233.32\\
160K & 13165.4 & 164.83 & 13065.6 & 246.52 & 13090 & 4943.32\\
320K & 15031.0 & 328.15 & 14911.6 & 516.66 & 14915 & 21167.50
\end{tabular}
}
\caption{Comparison of removal strategies for the Styria dataset. I is the size of the input (number of vertices). Columns 2+3: output size and runtime of the original approach (the
numbers differ from Table~\ref{tbl:runtimes} because $5$ new instances
were used). Columns 4+5: output size and runtime for the variant where $v$
is picked at random, but the best re-triangulation is chosen. Columns 6+7: output size and runtime when always the vertex removal with smallest $L_\infty$-distortion is performed.}
\label{tbl:heap_vs_greedy}
\end{table}

We have also implemented and tested
a more elaborate approach where among all candidates that can be removed,
we select the one whose removal causes the smallest $L_\infty$-perturbation
to $\baseterrain$. Similar to the previous variant,
it appears beneficial to remove vertices first
that cause small geometric distortion. This strategy requires, however, to initially check
all vertices of $\terrain$ to determine the best candidate.
Also, after every vertex removal, all neighbors have to be re-checked.
Table~\ref{tbl:heap_vs_greedy} shows that this strategy does not lead
to a noticeable decrease in the output size, but is substantially
increases the runtime. We remark that already initialization phase consumes
about 50\% of the running time of the randomized strategy.
Hence, we do not see a justification for this approach.

\begin{figure}[!htb]
\begin{center}
    \includegraphics[width=.47\textwidth]{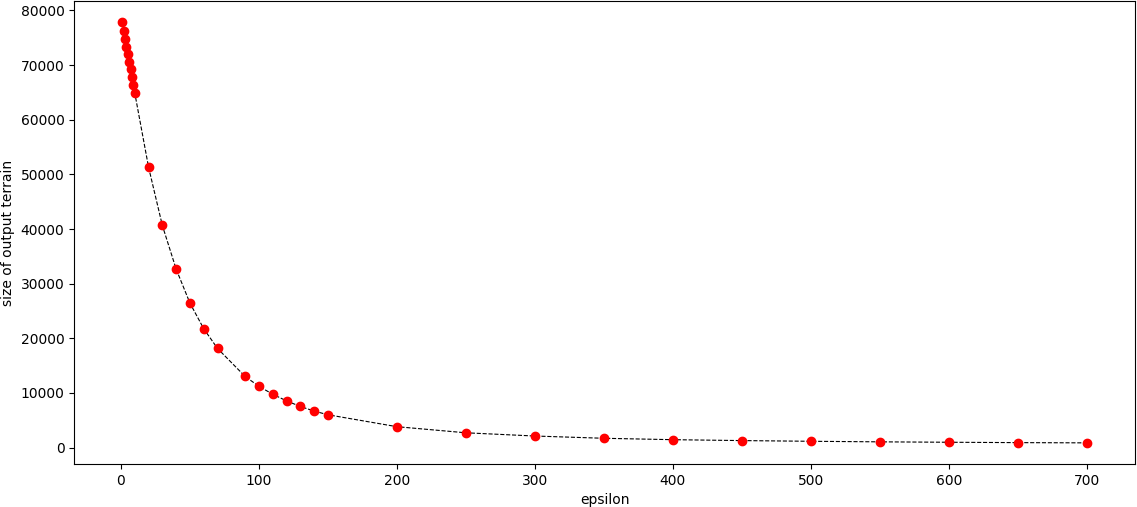}
\end{center}
    \caption{The size of the output terrain (number of vertices) for the $80$K input
dataset, depending on $\epsilon$. The number is averaged over $5$ runs.}
\label{fig:delta-plot}
\end{figure}

\paragraph{The effect of $\eps$}
The output size of our algorithm clearly depends on the chosen $\epsilon$.
In Figure~\ref{fig:delta-plot}, we display the size in dependence
on the values chosen. For very small values of $\epsilon$, the BLW algorithm
does not change the height of any vertex, edge, or triangle, meaning that
our subdivision strategy keeps the input terrain. For very large values
of $\epsilon$, the BLW algorithm removes all critical vertices, and
our reduction method can remove all interior vertices, only leaving the
boundary vertices. For $\epsilon$-values in between, we observe that
the size is monotonously decreasing. For instance, the terrain size
gets halved at around $\epsilon=30$ meters.

\section{Conclusions and discussion}\label{sec-conclusions}
We presented a method to reduce the total size of a topologically simplified
terrain, overcoming a major drawback of previous simplification methods
without giving up on its guarantees. We showed experimentally that the performance
is satisfying using the \textsc{Cgal} library.
We see room for further improvements, for instance,
by using better heuristics to avoid subdivisions in the BLW algorithm,
by using an inexact, yet numerically robust number type, and
by parallelizing our greedy removal procedure.

It is quite simple to see that our removal strategy does not always yield
the smallest possible terrain. While in practice, many vertices can be removed,
the complexity for finding an optimal reduction (for a fixed $\epsilon$)
is open. It would also be interesting to investigate how far our output is from
an optimal solution, both for the worst case and in expectation.

Since \textsc{Cgal}'s arrangement package is able to represent arrangements
on surfaces (such as spheres and tori), it will be simple
to extend our implementation to run on piecewise-linear functions on triangulated surfaces as well. We pose the question whether there are problems
for which a simplification method as described will be useful.

\section*{Acknowledgments}
Partially supported by the Austrian Science Fund (FWF) grant number P 29984-N35, by the Italian MIUR Award ``Dipartimento di Eccellenza 2018-2022'' - CUP: E11G18000350001, and by the SmartData@PoliTO center for Big Data and Machine Learning technologies. The authors thank Ulrich Bauer and Efi Fogel for helpful discussions.

\bibliographystyle{elsarticle-num}
\bibliography{TerrainBiblio}

\newpage
\appendix

\section{Proof details for persistence-aware triangulations}
\label{app:extended_proof}
We first show the converse of Theorem~\ref{thm:removal}.

\begin{theorem}\label{thm:converse-persistence-aware}
Let $v$ be a regular interior vertex of $\terrain$ of degree $d$. Let $\altterrain$ denote the terrain obtained by removing $v$ and re-triangulating its link
with $d-3$ diagonals. If $\terrain$ and $\altterrain$ have the same persistence diagram, all diagonals are persistence-aware.
\end{theorem}
\begin{proof}
By reductio ad absurdum, assume that an edge $ab$ is not persistence-aware.
That means, $a$ and $b$ are either both in the lower link or the upper link
of $v$. Let us consider the case first that both are in the upper link.
On the upper link path from $a$ to $b$, there must exist a vertex $c$
whose height value $t:=\terrainheight(c)$ is smaller than the height of $a$ and of $b$.
We claim that $\terrain$ and $\altterrain$ have different
homology in the sublevel set at scale $t-\eps$ or at $t$.
See Figure~\ref{fig:persistence_aware_converse_upper} for an illustration.
In general, $c$ is indeed not connected to any link vertex at scale $t$
in $\altterrain$, just because the edge $ab$ separates it from all vertices in the link with smaller height value.
On the other hand, $c$ is connected to $v$ in $\terrain$. Hence,
at scale $t$, the simplex-wise filtration of $\terrain$ adds exactly $c$ and the edge $cv$, whereas the simplex-wise filtration of $\altterrain$ only adds $c$ (and no edge within the link). Both filtrations might add
additional edges and triangles outside of the link, but these are the same
for both terrains.
It follows that the filtration of $\terrain$ adds exactly one more cell
at scale $t$ than the simplex-wise filtration of $\altterrain$. It follows
(by considering the Euler characteristic and the Euler-Poincar{\'e} formula)
that the homology groups differ either before $t$ or at $t$. This proves
the statement for the upper link case.

\begin{figure}[!htb]
    \centering
    \includegraphics[width=.47\textwidth]{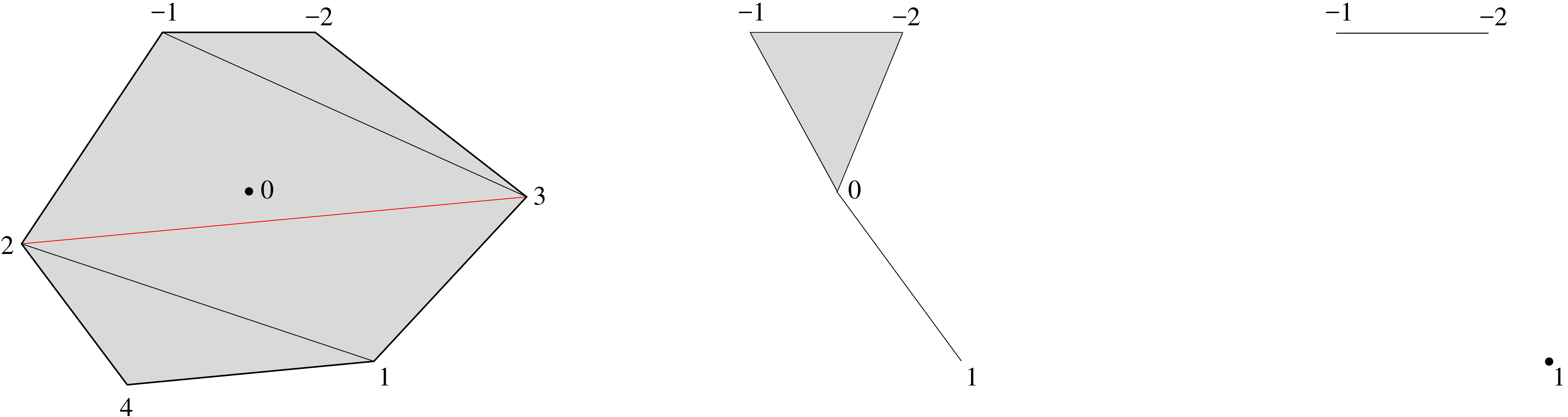}
    \caption{Illustration for the proof of Theorem~\ref{thm:converse-persistence-aware}, upper link case. Left: re-triangulation with a non-persistence-aware edge (in red). Middle: sublevel set of $1$ in $\terrain$. Right: sublevel set of $1$ in $\altterrain$.}
    \label{fig:persistence_aware_converse_upper}
\end{figure}

For the lower link case, let $c$ be the vertex on the lower link path
between $a$ and $b$ with maximal height $t:=\terrainheight(c)$.
Similar to the previous case,
we claim that the simplex-wise filtrations of $\terrain$ and $\altterrain$
add numbers of cells at height $t$ of different parity, implying
that the homology groups have to differ before or after $t$.
See Figure~\ref{fig:persistence_aware_converse_lower} for an illustration.

\begin{figure}[!htb]
    \centering
    \includegraphics[width=.47\textwidth]{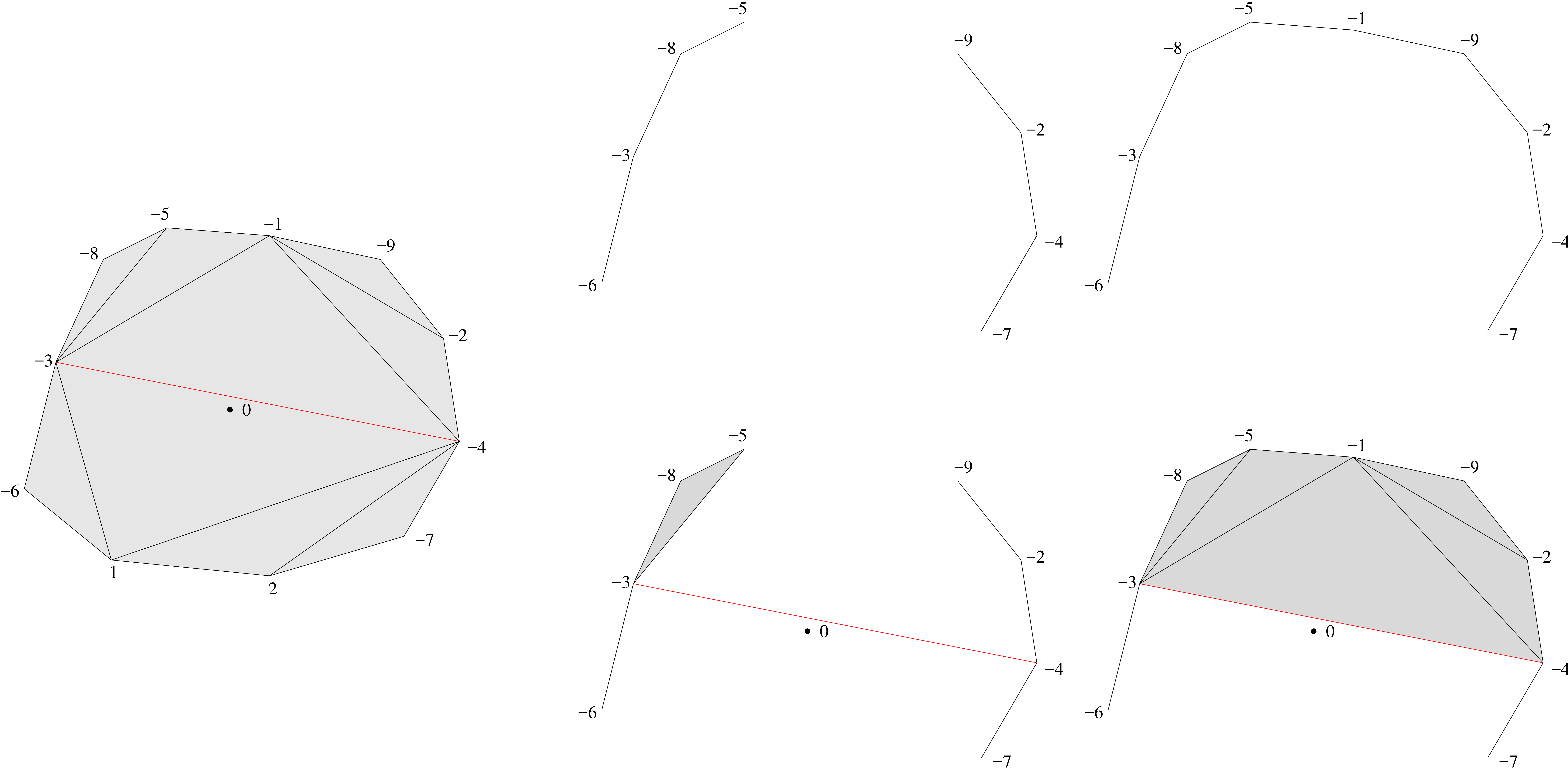}
    \caption{Illustration for the proof of Theorem~\ref{thm:converse-persistence-aware}, lower link case. Left: re-triangulation with a non-persistence-aware edge (in red). Upper row: sublevel sets of $-2$ and $-1$ in $\terrain$. Lower row: sublevel sets of $-2$ and $-1$ in $\altterrain$.}
    \label{fig:persistence_aware_converse_lower}
\end{figure}

Again, it suffices to count the number of cells added within the link,
since both terrains are identical outside the link. For $\terrain$,
at height $t$, the vertex $c$ and its two incident edges on the link
are added, resulting in $3$ cells in total.
In $\altterrain$, the non-persistence-aware edge $ab$ ensures
that $c$ is not adjacent to a link vertex with larger height.
It follows that, at height $t$, the filtration adds all edges incident
to $c$. Let $d$ denote its degree. Also triangles between consecutive
edges incident to $c$ get added, because $c$ is the maximal vertex of the
triangle. This yields $d-1$ triangles in total. It follows that $1$ vertex,
$d$ edges, and $d-1$ triangles are added, resulting in a total of $2d$ addition, which is an even number.
\end{proof}

A vertex of $\terrain$ is a \emph{minimum} if all its neighbors
have larger height, or equivalently, if its upper link is a cycle.
It is a \emph{maximum} if its upper link is empty.
It is \emph{regular} if its upper link is simply-connected.
It is a \emph{$k$-saddle} with $k\geq 2$ if its upper link consists
of $k$ connected components. Every vertex is either a minimum, maximum,
regular points, or $k$-saddle, and we call its type the
\emph{criticality type} of the vertex. Note that we consider
a $2$-saddle and a $3$-saddle to be of different criticality type.
The following statement shows the connection of our method with
the approach by Bajaj and Schikore~\cite{bs-topology}.

\begin{theorem}
Let $v$ be a regular interior vertex of $\terrain$ of degree $d$. Let $\altterrain$ denote the terrain obtained by removing $v$ and re-triangulating its link
with $d-3$ diagonals. All vertices of $\terrain$ (except $v$)
have the same criticality type in $\altterrain$ if and only if all diagonals are persistence-aware.
\end{theorem}
\begin{proof}
By inspection of Figures~\ref{fig:edge_flip_proof}
and \ref{fig:triangular_proof} that flipping a topologically flippable edge
and removing a regular vertex with a triangular link does not change
the criticality type of any vertex in the link. In the proof of
Theorem~\ref{thm:removal}, we showed that we can reduce the link
of a regular vertex to a triangle with flips of topologically flippable edges.
Hence, in this sequence, no vertex changes the criticality type, proving
one direction of the theorem.

For the other direction, assume that $ab$ is not persistence-aware.
If $a$ and $b$ are in the upper link, let $c$ be the vertex on the upper link
path between $a$ and $b$ with minimal height.
Let $u$ and $w$ denote the two neighbors of $c$ along the link of $v$.
In $\terrain$, the link of $c$
has a connected component in the lower link that consists only of $v$
because the two neighbors of $v$ on the link of $c$ are $u$ and $w$,
which are both higher than $c$ by the choice of $c$.
In $\altterrain$, $v$ is removed from the link of $c$, and possibly replaced
by a sequence of other vertices in the link of $v$. However, the presence
of the non-persistence-aware edge $ab$ ensures that all vertices that
become incident in $\altterrain$ have a larger height than $c$. It follows
that the lower link of $c$ loses a connected component from $\terrain$
to $\altterrain$, which means that either also the upper link loses a
component or the upper link becomes a cycle. In both cases, the vertex
$c$ changes its criticality type. The argument for a non-persistence-aware
lower link edge is symmetric.
\end{proof}

\section{Algorithmic details}
\label{app:algo_details}

\paragraph{Proof of Lemma~\ref{lem:overlay}}
If the first statement was wrong,
the only remaining possibilities are that a maximum occurs at a point $p$
in a face of $\baseterrain$, or at a point
where $uv$ overlaps with an edge of $\baseterrain$.
In both cases,
both $\simpterrainheight$ and $\baseterrainheight$ along $uv$ around $p$
are constant or linear functions.
There difference is either a linear function (which has
no local maximum) or constant, in which case the same height difference
is also attained at a vertex of either $\baseterrain$ or $\simpterrain$.

If the second statement was wrong, the only possibilities are that
the maximum occurs at a point $p$ on an edge of $\baseterrain$
or inside a face of it. In the former case, $\baseterrainheight$
is a constant or linear function along the edge,
and the same is true for $\simpterrainheight$.
This means that the difference is either a linear function
(contradicting the assumption of a maximum at $p$), or constant,
in which case the maximum is also achieved at a vertex of $\baseterrain$
or the boundary of the triangle. The case that $p$ is in a triangle
of $\baseterrain$ work in the same way, choosing an arbitrary line through $p$.

\paragraph{Placement of subdivision vertices}
We describe first the placement of a vertex of $\terrain$ that corresponds to
an edge $e=uv$ of $\baseterrain$. Let $\terrainheight_u$, $\terrainheight_v$, and $\terrainheight_e$ the function
values returned by the BLW algorithm for the vertices and the edge.
Assume, without loss of generality, that $\terrainheight_u\leq \terrainheight_v$.
By the filtration property,
also $\terrainheight_v\leq \terrainheight_e$ must hold.
With Lemma~\ref{lem:overlay}, it suffices to place the vertex
of $\terrain$ at point $p$ on the interior of the line segment such that
$f_e-\baseterrainheight(p)\leq \eps$, where $\baseterrainheight(p)$
denotes the height of $p$ in the terrain $\baseterrain$.
We try to find such a point $p$ that is as central
as possible between $u$ and $v$.

Note the following problem (also addressed in~\cite{aghlm-persistence}):
if $\terrainheight_e-\baseterrainheight(v)=\eps$ (that is, $e$ is pushed
upwards by the full $\eps$) and $\baseterrainheight(u)<
\baseterrainheight(v)$ (that is, the height along the edge is not constant),
we have that $f_e-\baseterrainheight(p)>\eps$ for every point $p$
in the interior of $uv$,
and a valid placement of the vertex for $e$ is impossible.
We thus assume, for simplicity, that $\eps$ is chosen so that
no persistence pair of $\baseterrain$ has persistence exactly $\eps$.
In that case, there exists $\eps'<\eps$ such that the BLW algorithm
applied on $\eps'$ still removes all persistence points of persistence
$\leq 2\eps$, and every cell changes its function value
by strictly less than $\eps$.

Now, we are looking for $\lambda\in (0,1)$ such that
\[\terrainheight_e-(\lambda\terrainheight_v+(1-\lambda)\terrainheight_u)\leq\eps.\]
A simple calculation yields the condition
\[
\lambda\geq \frac{\terrainheight_e-\eps-\terrainheight_u}{\terrainheight_v-\terrainheight_u}.\]
By assumption, $\terrainheight_e-\terrainheight_v<\eps$, hence the right hand
side is strictly smaller than $1$. We choose
\[\lambda_0:=\max\left\{\frac{1}{2},\frac{\terrainheight_e-\eps-\terrainheight_u}{\terrainheight_v-\terrainheight_u}\right\}.\]
and place the subdivision vertex for edge $e$ at
$\lambda_0 v + (1-\lambda_0)u$.

For a triangle $uvw$, assume that $w$ is the vertex with maximal height.
Let $m=\frac{u+v}{2}$ be the midpoint of $uv$. We then compute a
point $p$ in the interior of the line segment $wm$ (which lies in the interior
of the triangle), such that the distance of $\terrainheight(p)$
and $\terrainheight_e$ is at most $\eps$, with the same formula as above.

\paragraph{Triangular range queries}
Given three points $u$, $v$, $w$, we want
to report all vertices of $\baseterrain$
that are on or inside the triangle $uvw$.
We remark that \textsc{Cgal} offers an algorithm
for this problem
in the \emph{2D range and neighbor search package}~\cite{cgal-range}
which is based on Delaunay triangulations.
More precisely, this algorithm computes all points in the circumcircle
of $uvw$ and checks each encountered point
for being inside the triangle or not. This solution is unsatisfying for
very flat triangles, where the circumcenter becomes so large that
many false positives are found. In that case, the running time of the algorithm
depends on the number of vertices of $\baseterrain$, in theory and in practice.

Instead, we use the following approach:
initially, mark all vertices of $\baseterrain$ as unvisited.
Let $Q$ denote a queue that is initially empty.
For each boundary edge, we compute its zone
in the arrangement $\baseterrain$. The zone contains
all edges of $\baseterrain$ that cross the triangle boundary.
For every such edge, we check its endpoints whether they are unvisited
and in the triangle. If yes, we mark the vertex as visited,
put in into $Q$ and report the vertex.
For all vertices of $\baseterrain$ in the zone, we proceed in exactly the
same way.

Then, we pop elements from $Q$. For each vertex, we traverse its neighbors
in $\baseterrain$. When a neighbor is unvisited an in the triangle,
we mark it as visited, add it to $Q$ and report it. We terminate
when $Q$ is empty.

For the complexity of this approach, note that computing the zone
of a line segment has an (expected) complexity of $O(\log n+z)$,
where $z$ is the size of the zone, where the logarithmic factor
comes from point location for the first endpoint of the line segment
in $\baseterrain$. Denoting by $Z$ the sum of the zones of the three
triangle edges, the complexity is $O(\log n + Z + r)$, with $r$ the number
of vertices reported.

\section{Additional experimental results}
\label{app:exp_results}

Table \ref{tbl:runtimesTwo} presents the output size and the running time of our simplification procedure for uniform subsamples of the lake datasets. We always picked $\epsilon=100$ meters and we display the average execution time of 5 runs and the maximal deviation from the average.

\begin{table}[!hb]
\centering
\resizebox{0.49\textwidth}{!}{
\begin{tabular}{c||c|c|c|c|c}
 & I & S & C & O & T\\
\hline
\multirow{4}{*}{\rotatebox[origin=c]{90}{Garda}} & 10K & 14K & 111 & 2064.6 ($\pm 21.6$) & 8.62 ($\pm 0.07$)\\
& 20K & 29K & 136 & 2776.4 ($\pm 33.4$) & 19.42 ($\pm 1.16$)\\
& 40K & 58K & 122 & 3174.4 ($\pm 47.6$) & 40.28 ($\pm 0.73$)\\
& 80K & 113K & 91 & 3293.0 ($\pm 15.0$) & 82.13 ($\pm 1.64$)\\
\hline
\multirow{4}{*}{\rotatebox[origin=c]{90}{Como}} & 10K & 14K & 181 & 3195.2 ($\pm 9.8$) & 7.44 ($\pm 0.20$)\\
& 20K & 27K & 200 & 4239.4 ($\pm 27.6$) & 16.15 ($\pm 0.28$)\\
& 40K & 53K & 125 & 4775.2 ($\pm 39.8$) & 34.61 ($\pm 0.32$)\\
& 80K & 102K & 148 & 5088.4 ($\pm 21.4$) & 69.70 ($\pm 1.29$)\\
\hline
\multirow{6}{*}{\rotatebox[origin=c]{90}{Maggiore}} & 10K & 15K & 164 & 2551.8 ($\pm 19.8$) & 7.78 ($\pm 0.07$)\\
& 20K & 28K & 156 & 3284.0 ($\pm 18.0$) & 16.44 ($\pm 0.07$)\\
& 40K & 56K & 164 & 4076.6 ($\pm 60.4$) & 35.56 ($\pm 0.62$)\\
& 80K & 109K & 194 & 4581.4 ($\pm 27.4$) & 74.21 ($\pm 0.62$)\\
& 160K & 210K & 237 & 5062.8 ($\pm 28.8$) & 149.93 ($\pm 2.70$)\\
& 320K & 413K & 144 & 4796.8 ($\pm 51.8$) & 325.41 ($\pm 5.22$)
\end{tabular}
}
\caption{Benchmark results for the lake datasets. I is the size of the input, S is the size of the subdivision structure (the output of the BLW algorithm) using our improved subdivision strategy, C is the number of critical points of the BLW algorithm (and the output),  O is the output size of our method, and T is the running time of our method (in seconds). For O and T, we show the largest deviation from the average.}
\label{tbl:runtimesTwo}
\end{table}

\end{document}